\documentclass[journal,comsoc]{IEEEtran}
\IEEEoverridecommandlockouts

\usepackage{amssymb}
\usepackage{amsmath}
\usepackage{caption}
\usepackage{cite}
\usepackage{mathrsfs}
\usepackage[displaymath,mathlines]{lineno}
\usepackage{color}
\usepackage{overpic}
\usepackage{algorithm,algpseudocode}
\usepackage{algorithmicx}
\usepackage{pbox}
\usepackage{multicol}
\usepackage{amsthm}
\usepackage{epstopdf}

\makeatletter
\newcommand{\doublewidetilde}[1]{{%
		\mathpalette\double@widetilde{#1}%
}}
\newcommand{\double@widetilde}[2]{%
	\sbox\z@{$\m@th#1\widetilde{#2}$}%
	\ht\z@=.5\ht\z@
	\widetilde{\box\z@}%
}
\makeatother

\newtheorem{definition}{Definition}
\newtheorem{theorem}{Theorem}
\newtheorem{lemma}{Lemma}

\newtheorem{remark}{Remark}

\begin{document}
\title{\huge User Scheduling and Power Allocation for Precoded Multi-Beam High Throughput Satellite Systems with Individual Quality of Service Constraints}
\author{Trinh~Van~Chien, \textit{Member}, \textit{IEEE}, Eva~Lagunas, \textit{Senior Member}, \textit{IEEE}, Tung~Hai~Ta, Symeon~Chatzinotas, \textit{Senior Member}, \textit{IEEE}, and Bj\"{o}rn~Ottersten, \textit{Fellow}, \textit{IEEE}\\
\thanks{T. V. Chien, E. Lagunas, S. Chatzinotas, and B. Ottersten are with Interdisciplinary Centre for Security, Reliability and Trust (SnT), University of Luxembourg, Luxembourg (Email: vanchien.trinh@uni.lu, eva.lagunas@uni.lu, symeon.chatzinotas@uni.lu, and bjorn.ottersten@uni.lu).}
\thanks{T. H. Ta is with School of Information and Communication Technology (SoICT), Hanoi University of Science and Technology, Vietnam (Email: tungth@soict.hust.edu.vn).}
\thanks{This work has been partially supported by the Luxembourg National Research Fund (FNR) under the project FlexSAT ``Resource Optimization for Next Generation of Flexible SATellite Payloads" (C19/IS/13696663) and the European Space Agency (ESA) funded activity ``CGD - Prototype of a Centralized Broadband Gateway for Precoded Multi-beam Networks". The views of the authors of this paper do not necessarily reflect the views of ESA. Parts of this paper was presented at IEEE PIMRC 2021 \cite{trinh2021user}.}
}

\maketitle

\begin{abstract}
For extensive coverage areas, multi-beam high throughput satellite (MB-HTS) communication is a promising technology that plays a crucial role in delivering broadband services to many users with diverse Quality of Service (QoS) requirements. This paper focuses on MB-HTS systems where all beams reuse the same spectrum. In particular, we propose a novel user scheduling and power allocation design capable of providing guarantees in terms of the individual QoS requirements while maximizing the system throughput under a limited power budget. Precoding is employed in the forward link to mitigate  mutual interference at the users in multiple-access scenarios over different coherence time intervals. The combinatorial optimization structure from user scheduling requires an extremely high cost to obtain the global optimum even when a reduced number of users fit into a time slot. Therefore, we propose a heuristic algorithm yielding good trade-off between performance and computational complexity, applicable to a static operation framework of geostationary (GEO) satellite networks. Although the power allocation optimization is signomial programming, non-convex on a  standard form, the solution can be lower bounded by the global optimum of a geometric program with a hidden convex structure. A local solution to the joint user scheduling and power allocation problem is consequently obtained by a successive optimization approach. Numerical results demonstrate the effectiveness of our algorithms on large-scale systems by providing better QoS satisfaction combined with outstanding overall system throughput.
\end{abstract}
\IEEEpeerreviewmaketitle
\vspace*{-0.3cm}
\section{Introduction}
Mobile network generations, specially the latest $5$-th generation (5G) of cellular networks, have been designed as a response to the exponential growth of high data traffic and dense wireless devices \cite{kota20216g,tataria20216g}. Future generations, regarding beyond-5G evolution, have attempted to improve the system performance over prevalent wireless networks and further allow envisioned new applications in robotics, wireless security, and the internet of things (IoTs) comprising a massive number of heterogeneous devices \cite{giordani2020toward,chu2021robust}. By the use of revolutionary technologies as massive multiple-input multiple-output (MIMO) \cite{massivemimobook, 9531522}, mmWave communications \cite{hong2021role}, and network densification \cite{rodriguez2021secure}, terrestrial networks can handle multiple access scenarios by simultaneously serving many users with ubiquitous services, low latency, and high reliability \cite{yan2021scalable}. Nonetheless, terrestrial networks have mainly concentrated on urban and suburban areas but cannot remote regions and yield the poor coverage, for instance, oceans or high mountains, and especially under harsh environments with the presence of enormous obstacles \cite{fu2021collaborative,abdu2021flexible}. The systems will be strenuous to handle many users with heterogeneous services.  Under the broad geographic coverage, satellite technologies have effectively controlled this matter and provided global wireless services. Hence, satellite communications are being investigated by several major key network operators and  related  standardization  bodies including the third  generation  partnership  project  (3GPP), with its non-terrestrial  networks for 5G-and-beyond wireless systems that has been initially  introduced  in Release~$15$ \cite{3gpp2019study}.  
  
Multi-beam high throughput satellite (MB-HTS) systems are known to provide ubiquitous high-speed services of universal access to many users in a large coverage area that is inaccessible, insufficient, and expensive places with conventional terrestrial networks \cite{9210567,8746876, abdu2021flexible, zhang2019spectrum,ivanov2020spatial}. Unlike mono-beam satellites, the received signal strength can be increased thanks to new antenna architectures that are able to conform narrow beam spots on the Earth, resulting in high beamforming gains and spatially multiplexed communications. MB-HTS systems can provide significant improvements in the instantaneous throughput to concurrently support massive number of users with individual rate demands \cite{hofmann2019direct,7765141, hofmann2019ultranarrowband}. The multi-spot beams enable an MB-HTS system to offer more service flexibility to satisfy heterogeneous demands from multiple users sharing the same time and frequency resource while maintaining inter-beam interference at  acceptable levels. To further upgrade the channel capacity, most MB-HTS systems operate in the Ka-band, moving the feeder links to much higher frequencies than those of the contemporary mobile networks leading to seamlessly fast broadband connectivity across a large coverage area \cite{bao2021cooperative}. The performance of MB-HTS systems with aggressive frequency reuse heavily depends on both the precoding design and the user scheduling mechanism, which should be jointly optimized to obtain the globally optimal performance due to its coupled nature as pointed in \cite{vazquez2016precoding, bandi2019joint}. Joint optimization is extraordinarily challenging for practical systems since precoding coefficients are chosen based on the scheduled users' channel state information (CSI), and the scheduled users' performance depends on the precoding design, thus ended up with a very complex procedure. De facto, a system performance close to the optimal can be attained when users with semi-orthogonal channel vectors are selected \cite{Yoo2006a,8401547, ivanov2020spatial}. By fixing the precoding technique, most of the previous works have focused on the user scheduling designs for a single time slot by estimating the orthogonality between the channel vectors using, for example, the cosine similarity metric \cite{7091022} or the semi-orthogonality projection \cite{Yoo2006a}. However, the user scheduling over multiple time slots, i.e., block scheduling design, will be different and more challenging to maintain the QoSs of scheduled users due to various possible combinations conditioned on a large number of available users \cite{chen2020task}. To the best of the authors' knowledge, it is the first time MB-HTS block scheduling with individual QoS constraints and power control has been ever investigated in the literature. 

By considering an extended period with many time slots, this paper explores the benefits of block-based user scheduling in enhancing the system throughput whilst still making efforts to maintain the QoS requirement for each user in MB-HTS systems with full frequency reuse. Our main contributions are listed as follows:
\begin{itemize}
 \item We formulate a novel user scheduling problem spanning different time slots and whose goal is the sum-throughput maximization for an observed window time and the user-specific QoS constraints. Due to a combinatorial structure, the global optimum to this problem may be obtained by an exhaustive search of the parameter space, but discarded due to the impracticably arise from the exponential increase of the potential scheduling solutions if many users simultaneously request to access the network.
 \item We propose a heuristic algorithm that yields a local solution in polynomial time but can work for satellite systems providing service to a huge amount of user terminals, such as GEO satellite systems. For each time slot, the proposed algorithm schedules the available users conditioned on the monotonically non-decreasing sum throughput utility function and the individual QoSs. We also theoretically provide the convergence analysis and the computational complexity order required to operate such algorithm in practical systems. A flexible framework is also proposed allow the system to serve more users.
 \item We extend the scheduling framework to include the power coefficients as optimization variables. Despite a given scheduler-user set, the extended problem is a signomial program, whose global optimum is lower bounded by the solution of a geometric program with a hidden convex structure. Hence, the proposed heuristic scheduling algorithm and the successive optimization approach are exploited to obtain a stationary point.
 \item The proposed algorithms are evaluated by numerical results with a Defocused Phased Array-Fed Reflector Beam-Pattern provided by the European Space Agency (ESA) in the context of \cite{ESA}. Our solutions outperform some benchmarks on a long-term observation. The users' QoS requirements formulated with specific user data demands are satisfied with a high percentage. It is also shown that power allocation plays a crucial role in maintaining the QoS requirements of scheduled users.
 \end{itemize}
The remaining of this paper is organized as follows: Section~\ref{Sec:SysPer} presents in detail the system model and the net throughput per user over an observed window time. In  Section~\ref{Sec:SumThroughput}, under the QoS requirements at the scheduled users, the sum throughput optimization problem is formulated and solved for a given transmit power level at the satellite. An extension to jointly optimize the user scheduling and power allocation is presented in Section~\ref{Sec:Joint}. Finally, Section~\ref{Sec:NumericalResults} gives the numerical results, and the main conclusions are drawn in Section~\ref{Sec:Conclusion}.

\textit{Notation}: The upper and lower bold letters denote the matrices and vectors, respectively. The superscripts $(\cdot)^H$ and $(\cdot)^T$ are the Hermitian and regular transposes. The Euclidean norm is $\| \cdot \|$, $\mathrm{tr}(\cdot)$ is the trace of a matrix, and $\mathcal{CN}(\cdot, \cdot)$ is the circularly symmetric Gaussian distribution. The union of sets is  $\cup$ and $\subseteq$ denotes the subset operator. Finally, the cardinality of a set $\mathcal{B}$ is denoted as $|\mathcal{B}|$ and $\mathcal{O}(\cdot)$ is the big-$\mathcal{O}$ notation. 
\begin{figure*}[t]
   	\begin{minipage}{0.49\textwidth}
	   	\centering
	   	\includegraphics[trim=0.3cm 0.0cm 0cm 0cm, clip=true, width=3.4in]{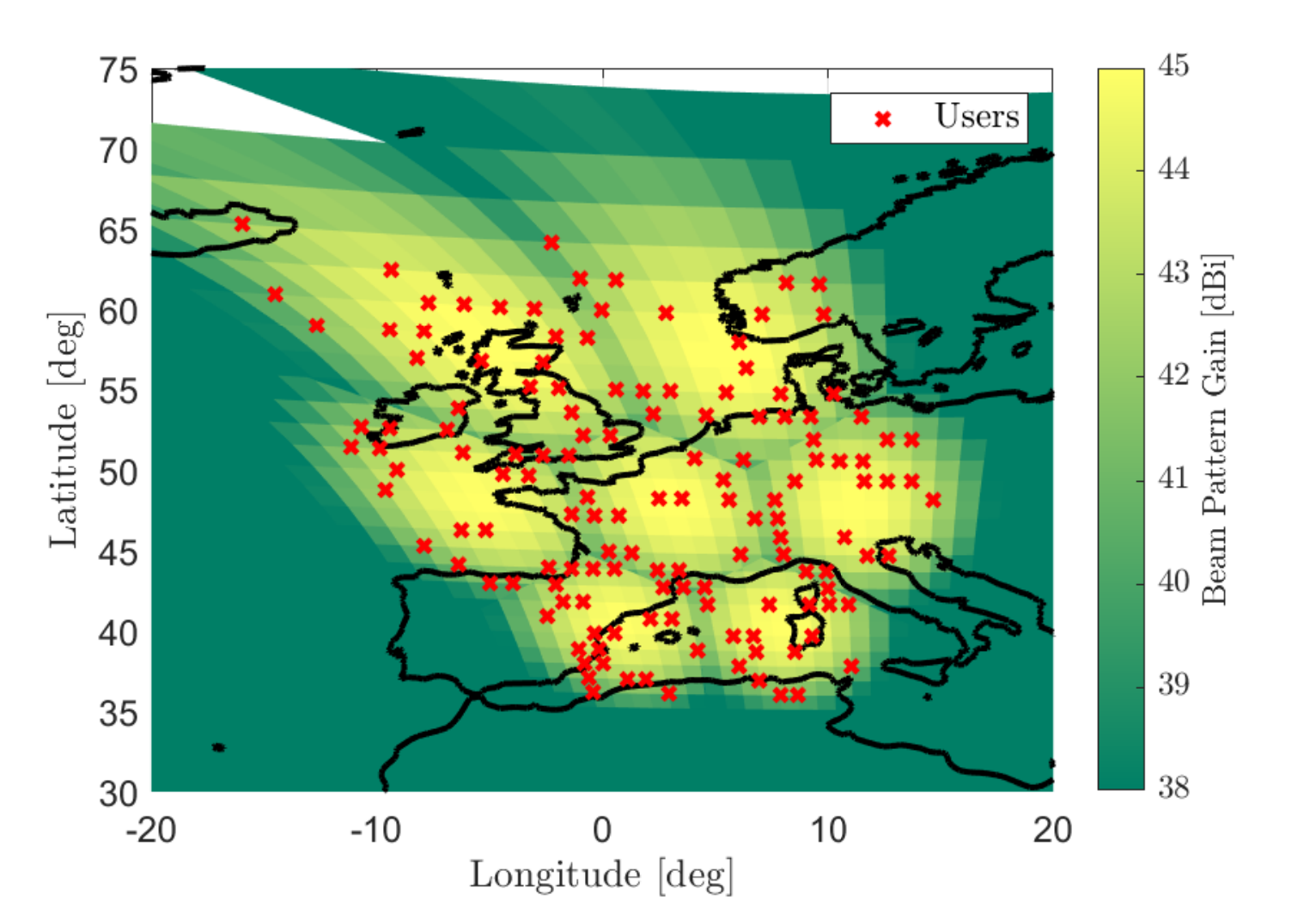} \vspace*{-0.1cm}\\
	   	$(a)$
    \end{minipage}
\hfill
    \begin{minipage}{0.49\textwidth}
    	\centering
    	\includegraphics[trim=4.7cm 7.0cm 16.45cm 3.5cm, clip=true, width=3.4in]{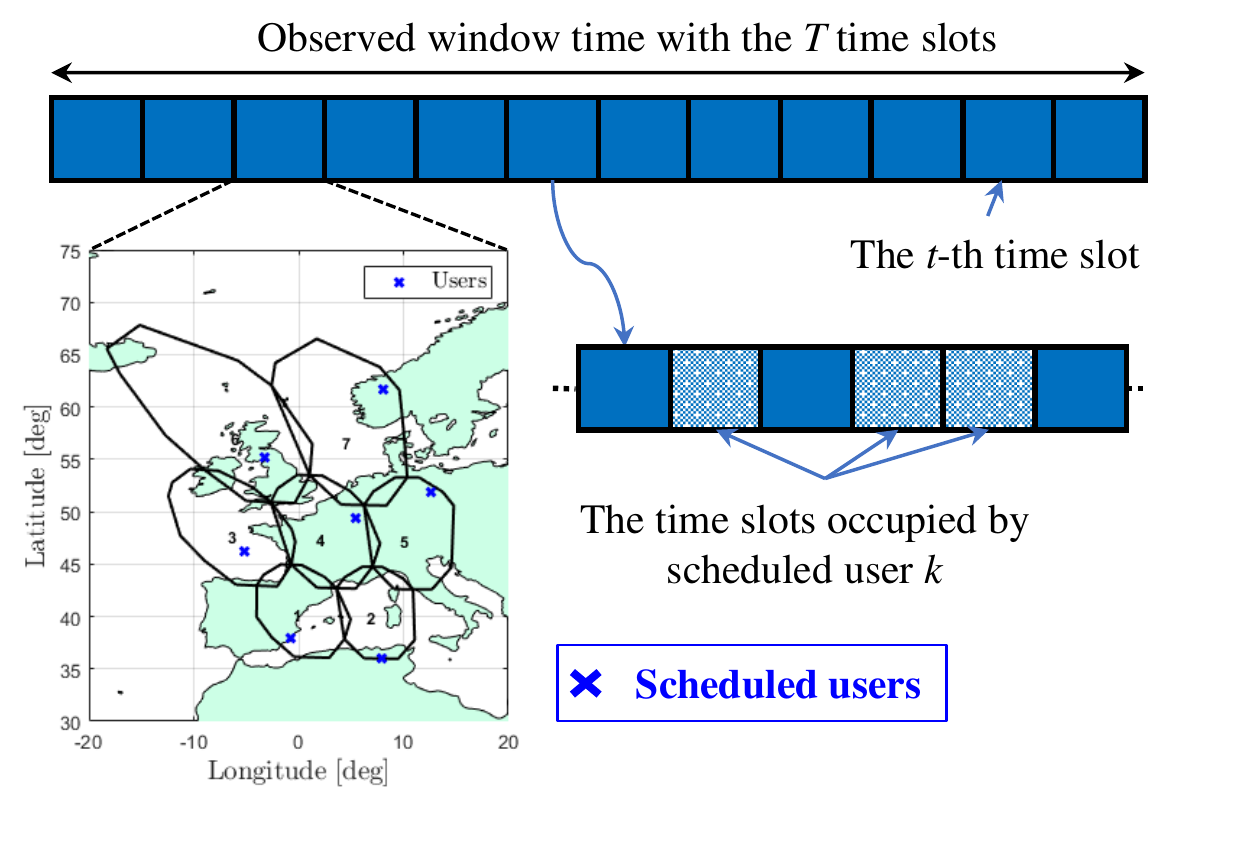} \vspace*{-0.1cm} \\
    	$(b)$
    \end{minipage}
	\caption{The considered MB-HTS system model with one GEO satellite and many available users in the coverage area: $(a)$ the beam pattern gain [dBi] and the users' positions; and $(b)$ The observed window time comprising the $T$ time slots each serving at most $M=7$ scheduled users.}
	\label{FigSysModel}
	\vspace*{-0.5cm}
\end{figure*}
\vspace*{-0.2cm}
\section{System Model \& Performance Analysis} \label{Sec:SysPer}
This section introduces an MB-HTS system model in which a single user per beam is scheduled at each time instance and  flexible in the sense of using multiple user links to serve users within the same beam coverage area  the number of scheduled users. Moreover, the system includes overlapping  beams depending on the beam-pattern and the antenna architecture. The aggregated and instantaneous downlink throughput for every scheduled user is then presented under the considered scheduling framework.
\vspace*{-0.25cm}
\subsection{System Model}
We consider the  downlink  of  a geostationary (GEO)  broadband  MB-HTS system that aggressively reuses the user link frequency. A precoding technique is implemented  to mitigate co-channel interference. The satellite is assumed to generate $M$ partially overlapping beams, as illustrated in Fig.~\ref{FigSysModel}(a). There are $N$ single-antenna users available with $N$ much larger than $M$ in the coverage area. In each beam, multiple users are multiplexed in a Time-Division Multiplexing (TDM) protocol in the forward link on a DVB-S2X carrier from a single gateway to these overlapping
beams \cite{9373415}. In addition, Time Division Multiple Access (TDMA) is used on the return link that enables low cost to obtain very accurate instantaneous channels. The beam pattern gain that influences the offered data throughput of each scheduled user is also illustrated.  We assume that the system performance is optimized in an observed window time comprising $T$ time slots. The system is supposed to operate in a unicast mode in which at most $M$ users can be scheduled per time slot, e.g., the blue users in Fig.~\ref{FigSysModel}(b). As the main distinction from related works \cite{7091022, deyi2021hybrid}, this paper considers the user scheduling problem  in a window time; thus, a scheduled user may occupy multiple time slots to satisfy its individual QoS requirement, as sketched in Fig.\ref{FigSysModel}(b). By focusing on the fixed-satellite service \cite{guidolin2015study}, user locations are geographically fixed, but the transmit data signals are independently distributed and mutually exclusive. Let us denote $\mathcal{K}(t)$ the scheduled-user set at the $t$-th time slot, which satisfies
\begin{equation}
\mathcal{K}(t) \subseteq \{1, \ldots, N \}  \mbox{ and } |\mathcal{K}(t)| \leq M.
\end{equation}
The propagation channels are
assumed to be static in the observed window time, which is a  valid assumption for GEO satellite
communications and reasonable window lengths, especially under clear sky. Specifically, if the channel between user~$k$ and the satellite is $\mathbf{h}_k \in \mathbb{C}^M$, then we can denote the channel matrix at the $t$-th time slot as
\begin{equation}
\mathbf{H}(t) = \big[\mathbf{h}_{\pi_1}, \ldots,\mathbf{h}_{\pi_{|\mathcal{K}(t)|}}  \big] \in \mathbb{C}^{M \times |\mathcal{K}(t)|},
\end{equation}
with $\pi_1, \ldots,\pi_{|\mathcal{K}(t) |} $  being the user indices in $\mathcal{K}(t)$. Subsequently, the size of the channel matrix depends on the cardinality $|\mathcal{K}(t)|$. Based on some practical channel features, $\mathbf{H}(t)$ is formulated as
\begin{equation}
\mathbf{H}(t) =   \mathbf{B}(t) \pmb{\Phi}(t),
\end{equation}
where $\mathbf{B}(t) \in \mathbb{R}_{+}^{M \times |\mathcal{K}(t)|}$ represents different aspects in satellite communications comprising the received antenna gain, thermal noise, path loss, and satellite antenna radiation pattern with the $(m,k)-$th element defined as
\begin{equation}
b_{mk} = \frac{\lambda \sqrt{\widehat{G}_{Rk} G_{mk}}}{ 4 \pi d_{mk}}, m=1, \ldots M, k = 1, \ldots, |\mathcal{K}(t)|,
\end{equation}
where $\lambda$ is the wavelength of a plane wave; $d_{mk}$ is the distance between the $m$-th satellite antenna and user~$k$. It is safe to assume $d_{1k} = \ldots = d_{Mk}, \forall k,$ for a GEO satellite system because of long propagation distance.  The receiver antenna gain is denoted as $\widehat{G}_{Rk}$, which mainly depends on the receiving antenna aperture, whilst $G_{mk}$ is the gain defined by the satellite radiation pattern and user location; $\pmb{\Phi}(t) \in \mathbb{C}^{|\mathcal{K}(t)| \times |\mathcal{K}(t)|}$ is a diagonal matrix denoting signal phase rotations with the  $(k,k)$-th diagonal element that is $\phi_{kk} = e^{i \psi_k}, \forall k = 1, \ldots, |\mathcal{K}(t)|,$ where  $\psi_k$  is the user-related phase, which is identically and independently distributed from the satellite payload.
\vspace*{-0.25cm}
\subsection{Downlink Data Transmission}
At the $t$-th time slot, the satellite is simultaneously transmitting data signals to the scheduled users. In detail, $s_k(t)$ is the modulated data symbol for scheduled user~$k$ with $ |s_k(t)|^2 = 1$. 
The received signal at scheduled user $k$ with $k \in \mathcal{K}(t)$, denoted by $y_k(t) \in \mathbb{C}$, is thus formulated as 
\begin{equation} \label{eq:ReceivedSig}
	y_k(t) = \sum\limits_{k' \in \mathcal{K}(t)} \sqrt{p_{k'} (t)} \mathbf{h}_{k}^H \mathbf{w}_{k'}(t) s_{k'}(t) + n_k(t),
\end{equation}
where $\mathbf{w}_k(t)$ is the precoding vector used for scheduled user~$k$ with $\| \mathbf{w} (t) \| = 1$ and $p_k(t)$ is data power allocated to this user at the $t$-th time slot; $n_k(t)$ is additive noise with $n_k(t) \sim \mathcal{CN}(0,\sigma^2)$ and $\sigma^2$ being the noise variance. Although the channels are static in the observed window time, the precoding vectors vary upon time slots due to the user scheduling.
The limited power budget at the satellite can be expressed as
\begin{equation}
	\sum\limits_{k' \in \mathcal{K}(t)} p_{k'}(t) \| \mathbf{w}_{k'} (t) s_{k'}(t) \|^2 = 	\sum\limits_{k' \in \mathcal{K}(t)} p_{k'}(t) \leq P_{\max},
\end{equation}
where $P_{\max}$ is the maximum transmit power that the satellite can spend  for data symbols at the $t$-th time slot. In order to compute the instantaneous throughput of scheduled user~$k$, we recast the received signal \eqref{eq:ReceivedSig} into an equivalent form as
\begin{equation} \label{eq:ReceivedSigv1}
	\begin{split}
	 y_k(t) =& \sqrt{p_{k}(t)} \mathbf{h}_{k}^H \mathbf{w}_{k}(t) s_{k}(t) +  \sum\limits_{k' \in \mathcal{K}(t) \setminus \{ k\} } \sqrt{p_{k'}(t)}  \times\\
	 & \mathbf{h}_{k}^H \mathbf{w}_{k'}(t)  s_{k'}(t) + n_k(t),
	\end{split}
\end{equation}
where the first part contains the desired signal, while the second part is mutual interference from the other scheduled users at the $t$-th time slot. From \eqref{eq:ReceivedSigv1}, the aggregated and per-time-slot throughput of scheduled user~$k$ is given in Lemma~\ref{lemma:ChannelCapacity}.
\begin{lemma} \label{lemma:ChannelCapacity}
Assuming that user~$k$ is scheduled only in the $T_k$ time slots, $1 \leq T_k \leq T$, its aggregated throughput is
\begin{equation} \label{eq:Rk}
	R_k\left( \{ \mathcal{A}(t) \} \right) = \sum\limits_{t=1}^{T_k} R_k(\mathcal{A}(t) ),\mbox{[Mbps]},
\end{equation}
where $\mathcal{A}(t) = \mathcal{K}(t) \cup \{ p(t) \}$ and $R_k(\mathcal{A}(t) )$ is the instantaneous throughput at the $t$-th time slot, $1\leq t \leq T_k$, which is computed as 
\begin{equation}\label{eq:Capacityk}
R_k(\mathcal{A}(t) ) = B \log_2 \left( 1 + \mathrm{SINR}_k(\mathcal{A}(t) ) \right), \mbox{[Mbps]},
\end{equation}
where $B$~[MHz] is the system bandwidth and the signal-to-interference-and-noise ratio is
\begin{equation} \label{eq:SINRk}
\mathrm{SINR}_k(\mathcal{A}(t))  = \frac{ p_{k} (t) \big| \mathbf{h}_{k}^H \mathbf{w}_{k}(t) \big|^2 }{ \sum\limits_{k' \in \mathcal{K}(t) \setminus \{ k\} } p_{k'} (t) \big|\mathbf{h}_{k}^H \mathbf{w}_{k'}(t)|^2 +\sigma^2}.
\end{equation}
\end{lemma}
\begin{proof}
The instantaneous throughput of scheduled user~$k$ at each time slot is computed as \eqref{eq:Capacityk} by exploiting the Shannon channel capacity under perfect channel state information and known mutual interference. The aggregated throughput is further accumulated over all the $T_k$ time slots as in  \eqref{eq:Rk}.
\end{proof}
As a consequence of the practical frameworks, the instantaneous throughput in \eqref{eq:Capacityk} measures an portion of data rate that the network can offer at the current time slot. The aggregated throughput represents the total data rate in all scheduling instances that scheduled user~$k$ accesses the network. In more detail, for a given transmit power coefficients, the instantaneous throughput in \eqref{eq:Capacityk} is a function of both the scheduled-user set $\mathcal{K}(t)$ and the transmit power coefficients $\{ p_k (t) \}$, while the aggregated throughput depends on all the scheduled users in the $T_k$ time slots. It is noteworthy that the throughput in Lemma~\ref{lemma:ChannelCapacity} can be applied for arbitrary channel models and precoding techniques. This paper exploits linear precoding processing because it has a lower cost than the optimal solution, which has been  widely accepted in the satellite communications community \cite{abdu2021flexible}. More specifically, we deploy the regularized zero forcing (RZF) precoding matrix $\mathbf{W}(t) \in \mathbb{C}^{M \times |\mathcal{K}(t)|}$, which is\footnote{ Other popular linear prececoding techniques, such as maximum ratio transmission (MRT) and zero-forcing (ZF), focus on maximizing the signal gain or canceling out mutual interference. In multiple access scenarios with the coexistence of additive noise and interference from sharing radio resources, an optimal linear signal processing technique should take advantages of both the MRT and ZF precoding techniques as what has been achieved by the RZF precoding technique. Previous works, e.g., \cite{christopoulos2012linear}, reported  that the RZF precoding technique offers good system capacity and is comparable to the optimal solution in many scenarios.}
\begin{equation} \label{eq:MMSEPrecodMatrix}
\mathbf{W}(t) = \frac{1}{\sqrt{\omega(t)}} \mathbf{H}(t) \left(\mathbf{H}(t)^H \mathbf{H}(t) + \frac{\sigma^2|\mathcal{K}(t)|}{P_{\max}} \mathbf{I}_{ |\mathcal{K}(t)|} \right)^{-1},
\end{equation}
where $\mathbf{I}_{ |\mathcal{K}(t)|}$ is the identity matrix of size $|\mathcal{K}(t)| \times |\mathcal{K}(t)|$ and the normalized power constant $\gamma(t)$ is defined as
\begin{equation}
\omega(t) = \mathrm{tr} \left( \mathbf{H}(t) \Big(\mathbf{H}^H(t) \mathbf{H}(t) + \frac{\sigma^2 |\mathcal{K}(t)|}{P_{\max}} \mathbf{I}_{|\mathcal{K}(t)|} \Big)^{-2} \mathbf{H}^H(t) \right).
\end{equation}
We should notice that each precoding matrix in \eqref{eq:MMSEPrecodMatrix} is a function of the scheduled users at the $t$-th time slot, thus it verifies the high importance of a proper set $\mathcal{K}(t)$ in boosting the throughput. By counting for the arithmetic operations with the high cost such as complex multiplications and divisions \cite{van2019power}, the computational complexity order to construct an RZF precoding matrix is presented in Lemma~\ref{lemma:ComplexityWtv2}. 
\begin{lemma} \label{lemma:ComplexityWtv2}
 For a given set of the transmit power coefficients $\{ p_k(t)\}$ and channel matrix, the precoding matrix $\mathbf{W}(t)$ is constructed by the computational complexity in the order of $\mathcal{O}\big( \frac{1}{2}M^2 |\mathcal{K}(t)| \big)$ as a consequence of the channel matrix $\mathbf{H}(t)$ depending on the scheduled-user set $\mathcal{K}(t)$.
\end{lemma}
\begin{proof}
By applying \cite[Lemma~B.1]{massivemimobook} to the channel matrix $\mathbf{H}(t)$, the product $\mathbf{H}^H(t)\mathbf{H}(t)$ requires $\frac{1}{2}|\mathcal{K}(t)|(|\mathcal{K}(t)|+1)M$ complex multiplications thanks to the Hermitian symmetry. Let us introduce a new matrix 
\begin{equation}
\mathbf{G}(t) = \mathbf{H}(t)^H \mathbf{H}(t) + \frac{\sigma^2 |\mathcal{K}(t)|}{P_{\max}} \mathbf{I}_{|\mathcal{K}(t)|},
\end{equation}
then attaining $\mathbf{G}(t)$ needs $\left( \frac{1}{2}M(|\mathcal{K}(t)|+1) + 1\right)|\mathcal{K}(t)|$ complex multiplications. According to \cite[Lemma~B.2]{massivemimobook}, the inverse matrix $\mathbf{H}(t)\mathbf{G}^{-1}(t)$ can be computed efficiently by utilizing the Cholesky decomposition that includes $\frac{|\mathcal{K}(t)|^3 - |\mathcal{K}(t)|}{3} + |\mathcal{K}(t)|^2 M + |\mathcal{K}(t)|$ complex multiplications and divisions. Furthermore, we need the $\frac{1}{2}(M^2 +M)|\mathcal{K}(t)| + 2$ complex multiplications, division, and square root to obtain $\gamma(t)$. Thus, the number of the arithmetic operations to obtain the RZF precoding matrix $\mathbf{W}(t)$ is obtained by adding all the cost. 
Due to the fact $|\mathcal{K}(t)| \leq M$, we can ignore the terms with low degree in the obtained posynomial and hence the computational complexity order is shown as in the lemma.
\end{proof}
The key achievement from Lemma~\ref{lemma:ComplexityWtv2} is to point out the computational complexity of the RZF precoding matrix construction directly proportional to the total elements in the scheduled-user set $\mathcal{K}(t)$ for a given number of satellite beams. We later utilize Lemma~\ref{lemma:ComplexityWtv2} to evaluate the complexity order of the proposed algorithm to the user scheduling problem. 

From \eqref{eq:MMSEPrecodMatrix}, the precoding vector dedicated to scheduled user~$k$  at each time slot, i.e., $\mathbf{w}_k(t)$, is selected as the $k-$th column of matrix $\mathbf{W}_k(t)$. By exploiting a similar methodology as what has done for Lemma~\ref{lemma:ComplexityWtv2}, it is straightforward to manifest that RZF precoding has the higher computational complexity than other linear signal processing techniques such as MRT or ZF. Nonetheless, this precoding selection provides better throughput than the others and avoiding an ill-posed inverse appearing when the channels are highly correlated leading to rank deficiency.
\vspace*{-0.25cm}
\section{User Scheduling for Sum Throughput Maximization With Fixed Power Level} \label{Sec:SumThroughput}
By considering the user scheduling over many different time slots, a sum throughput optimization problem with the QoS requirements is  formulated for a fixed power level. Because of the inherent non-convexity, a heuristic algorithm is then proposed to obtain a local solution in polynomial time.
\vspace*{-0.25cm}
\subsection{Problem Formulation}
Our objective function in this paper is the total sum throughput of all the scheduled users in the considered window time and the individual QoS requirements of scheduled users are constraints. By fixing the power coefficients, it leads to $\mathcal{A}(t) = \mathcal{K}(t)$. Hence, the optimization problem, which we would like to solve, is mathematically formulated as
\begin{subequations} \label{Problemv1}
	\begin{alignat}{2}
		& \underset{\{ \mathcal{K}(t) \} }{\mathrm{maximize}}
		& & \, \sum\limits_{t=1}^T \sum\limits_{k \in \mathcal{K}(t) } R_k(\mathcal{K}(t)) \\
		& \mbox{subject to}
		& &   \, R_k(\mathcal{K}(t)) \geq \xi_k /T_k, \forall k \in \mathcal{K}(t), \forall t, \label{eq:QoSconst}\\
		&&& \, \mathcal{K} (t) \subseteq \{ 1, \ldots, N\}, \forall t, \label{eq:Kt1}\\
		&&& \, | \mathcal{K} (t)| \leq M, \forall t, \label{eq:Kt2} \\
		&&& \,  \cup \mathcal{K}(t) \subseteq \{1, \ldots,N \} \label{eq:Kt3},
	\end{alignat}
\end{subequations}
where $T_k$ is the number of time slot that spends on scheduled user~$k$ to fulfill the QoS requirement, denoted by $\xi_k$ [Mb] as in \eqref{eq:QoSconst}. As $T$ is sufficiently large, the long-term QoS satisfaction of user~$k$ is defined as $R_k(\{ \mathcal{K}(t) \}) \geq \xi_k T_k$, which is spontaneously fulfilled when all the per-time-slot constraints in \eqref{eq:QoSconst} hold. Furthermore, \eqref{eq:Kt1}--\eqref{eq:Kt3} show the conditions on all the scheduled-user sets $\mathcal{K} (t), \forall t$. Specifically, \eqref{eq:Kt1} implies that every $\mathcal{K}(t)$ is a subset of the available-user set, say $\{1,\ldots,N\}$, whilst \eqref{eq:Kt2} implies that the number of scheduled users may be less than the available beams to maximize the sum throughput in the entire network and therefore demonstrating the flexibility of our optimization problem. The union of all the scheduled-user sets $\mathcal{K}(t), \forall t,$ over the observed window time is a subset of the available-user set in general. From the system viewpoint, some users may be ignored from service due to, for example, bad channel conditions and/nor too high QoS requirements such that they are not be served with a limited transmit power level.
\vspace*{-0.25cm}
\subsection{Problem Structure}
We stress that problem~\eqref{Problemv1} is non-convex as a consequence of the discrete feasible domain and the non-convex objective function. Particularly, the discrete feasible domain makes \eqref{Problemv1} a combinatorial problem, where the global optimum can only be obtained for a small scale network  with few users and small number of beams since an exhaustive search of the parameter space is required. Nevertheless, the exhaustive search has the computational complexity scaling up exponentially with the number of available users. For instance, with $M=7, N=100$, and only one time slot is considered for the sake of simplicity, the optimal solution is obtained by searching over the  following different combinations
\begin{equation}
\sum\limits_{k=1}^M \frac{N!}{k!(N-k)!} \approx 1.7 \times 10^{10},
\end{equation}
which is prohibitively large. An exhaustive search is, therefore, not preferable for large-scale networks with many users as the main consideration in this paper. For now, we  differentiate our user scheduling optimization problem from the related works as shown in Remark~\ref{Remark1}.
\begin{remark} \label{Remark1}
Problem~\eqref{Problemv1} is a generalized version of the previous works  \cite{Yoo2006a,Honnaiah2020} and references therein since the $N$ users are scheduled over different time slots and since we also take the QoS requirements into account. In other words, problem~\eqref{Problemv1} ensures the scheduled users always satisfied their throughput demand. Furthermore, an effective RZF precoding matrix constructed from a good scheduling scheme not only reduces mutual interference but also ameliorates the received signal strength that boost the system performance. With a limited window time and the correlation among propagation channels, the number of scheduled users might be less than the total available users to maximize the network throughput.
\end{remark}
	\begin{algorithm}[t]
	\caption{A user scheduling algorithm for problem~\eqref{Problemv1}} \label{Algorithm1}
	\textbf{Input}: Available-user set $\mathcal{N}(0) \leftarrow \{1, \ldots, N \}$; Scheduled-user set $\mathcal{K}(0) \leftarrow \emptyset$; Propagation channel vectors $\{\mathbf{h}_1, \ldots, \mathbf{h}_N \}$; QoS requirements $\{ \xi_1, \ldots, \xi_k \}$; Number of time slots $T$ and individual scheduled time slots $\{T_1, \ldots, T_N\}$; Transmit data powers $\{ p_1, \ldots, p_{N} \}$.
	\begin{itemize}
		\item[1.] Select scheduled user~$\pi_1$ based on the best channel gain as obtained in \eqref{eq:ChannelOrder}.
		\item[2.] Set $t=1$, then update $\mathcal{N}(1)$ and $\mathcal{K}(1)$ by \eqref{eq:1stUpdate}.
		\item[3.] \textbf{while} $t \leq T$ \textbf{do}
		\begin{itemize}
			\item[3.1.] Set $m = |\mathcal{K}(t-1)|$ and  $\mathcal{K}_m(t) = \mathcal{K}(t-1)$.
			\item[3.2.] \textbf{while} $m \leq M$ \textbf{do}
			\begin{itemize}
				\item[3.2.1.] Obtain user~$k_{m}^{t,\ast}$ and $\widetilde{\mathcal{K}}_m^{\ast}(t)$  by  solving problem \eqref{Prob:Ratev1} with $\widetilde{\mathcal{K}}_m(t)$ updated in \eqref{eq:Ktilde}.
				\item[3.2.2.] \textbf{If} the conditions \eqref{eq:Ser1} and \eqref{eq:Ser2} satisfy: Update $\mathcal{N}(t)$ and $\mathcal{K}_m(t)$ as \eqref{eq:Updatev1}. \textbf{Otherwise} keep $\mathcal{N}(t)$ and $\mathcal{K}_m(t)$ unchanged and go to Step $3.2.3.$
				\item[3.2.3.] Set $m=m+1$.
			\end{itemize}
			\item[3.3.] \textbf{end while}
			\item[3.4.] Update $\mathcal{K}(t)$ by \eqref{eq:Kt} and compute the throughput of scheduled users by \eqref{eq:Capacityk}.
			\item[3.5.] Find the scheduled users satisfied their QoS requirements (set $\widehat{\mathcal{K}}(t)$) by computing the aggregated throughput using \eqref{eq:Rk} and checking the condition \eqref{eq:Conditionv1}, then remove them from service by using \eqref{eq:Ktv1}.
			\item[3.6.] Update $\mathcal{K}(t)$ by \eqref{eq:Ktv1} and set $t=t+1$.
		\end{itemize}
		\item[4.] \textbf{end while}
	\end{itemize}
	\textbf{Output}: The scheduled user sets.
	\vspace*{-0.0cm}
\end{algorithm}
\vspace*{-0.25cm}
\subsection{User Scheduling Algorithm with Fixed Power Level} \label{Sec:2StageAlgo}
Motivated by large-scale networks with many users simultaneously requesting to admit the system, we propose a heuristic algorithm that obtains a good local solution in polynomial time with tolerable computational complexity. Algorithm~\ref{Algorithm1} demonstrates the proposal with the double loops: The outer loop indicates the evolution of time slots and the inner loop is for the growth of the scheduled users per time slot.  At the initial stage, let us denote $\mathcal{N}(0) \leftarrow \{1, \ldots, N\}$ the set of available users with the corresponding channels $\mathbf{h}_1, \ldots, \mathbf{h}_{N}$. Moreover, the scheduled user set $\mathcal{K}(0)$ is initially setup as an empty set. The proposed heuristic algorithm begins with sorting the channel gains in a descending order as
\begin{equation} \label{eq:ChannelOrder}
	\| \mathbf{h}_{\pi_1}  \|^2 \geq \| \mathbf{h}_{\pi_2}  \|^2 \geq \ldots \geq \| \mathbf{h}_{\pi_N}    \|^2,
\end{equation}
where $\{\pi_1, \ldots,  \pi_N \}$ is a permutation of the user indices for which \eqref{eq:ChannelOrder} holds. Then, we set the outer iteration index $t=1$ and the available- and scheduled-user sets are updated as
\begin{equation} \label{eq:1stUpdate}
	\mathcal{N}(1) \leftarrow \mathcal{N}(0) \setminus \{ \pi_1 \} \mbox{ and } \mathcal{K}(1) \leftarrow \mathcal{K}(0) \cup \{ \pi_1 \}.
\end{equation}
At the $t$-th outer iteration ($1\leq t \leq T$), if the number of scheduled users from the previous time slot, which have not been satisfied their QoS requirements yet, is less than the number of beams, i.e., $|\mathcal{K}(t-1)| < M$, there is room for scheduling new users to join the system if all the constraints of problem~\eqref{Problemv1} are satisfied. For such, an inner loop is implemented to testify whether or not at most the $M- |\mathcal{K}(t-1)| +1$~potential users can be scheduled. The following optimization problem is therefore considered at the $m-$th inner iteration ($|\mathcal{K}(t-1)| \leq m \leq M$):
\begin{equation} \label{Prob:Ratev1}
k_m^{t,\ast} = \underset{k \in \mathcal{N}(t)}{\mathrm{argmax}}  \sum\limits_{k' \in \widetilde{\mathcal{K}}_m(t) } R_{k'} \big(\widetilde{\mathcal{K}}_m(t) \big),
\end{equation}
where each set $\widetilde{\mathcal{K}}_m(t)$ is related to one user $k \in \mathcal{N}(t)$, which is defined as
\begin{equation} \label{eq:Ktilde}
\widetilde{\mathcal{K}}_m(t) \leftarrow  \begin{cases}
	\mathcal{K}_{m-1}(t) \cup \{k\}, & \mbox{if } m =  |\mathcal{K}(t-1)| +1, \ldots, M,\\
 \mathcal{K}(t-1) \cup \{k\}, & \mbox{if }  m =  |\mathcal{K}(t-1)|.
\end{cases}
\end{equation}
In \eqref{eq:Ktilde},  $\mathcal{K}_{m-1}(t) $ is the scheduled-user set at the $(m-1)-$th inner iteration with $\mathcal{K}_m (t) = \mathcal{K}(t-1)$ when $m=|\mathcal{K}(t-1)|$. Problem~\eqref{Prob:Ratev1} aims at maximizing the total sum throughput at a particular time slot only.\footnote{The solution to problem~\eqref{Prob:Ratev1} is not unique in general. Alternatively, there may be more than one user with the same total sum throughput, but we can select one of them for further processing.} Hence, the solution to problem \eqref{Prob:Ratev1} does not guarantee a monotonic increasing property, which is in need to have a good local solution to the original problem~\eqref{Problemv1}. As foreseen from a multi-user system, user~$k_m^{t,\ast}$ causes more mutual interference to other users in the set $\widetilde{\mathcal{K}}_m(t)$ that may lead to their throughput no longer satisfy the QoS requirements. In order to get rid of this issue, we suggest a mechanism to further testify whether or not user~$k_m^{t,\ast}$ becomes a scheduled user as in Theorem~\ref{Theorem:SelectedUser}.
\begin{theorem} \label{Theorem:SelectedUser}
 User~$k_m^{t,\ast}$ becomes a scheduled user if the following conditions satisfy
\begin{align}
  \sum\limits_{k' \in \widetilde{\mathcal{K}}_m^\ast (t) } R_{k'} \big( \widetilde{\mathcal{K}}_m^\ast (t) \big) &\geq \sum\limits_{k' \in \mathcal{K}_{m-1}(t) } R_{k'} \big(\mathcal{K}_{m-1}(t) \big), \label{eq:Ser1}\\
 R_{k'}(t) &\geq \frac{\xi_{k'}}{T_{k'}}, \forall k' \in \widetilde{\mathcal{K}}_m^{\ast}(t), \label{eq:Ser2}
\end{align}
where $\widetilde{\mathcal{K}}_m^{\ast}(t)$ is formulated as in \eqref{eq:Ktilde}, but for user~$k_m^{t,\ast}$. The condition~\eqref{eq:Ser1} guarantees the objective function of problem~\eqref{Problemv1} to be non-decreasing along with iterations until reaching a fixed point, while all users admitted to the network satisfy their QoS requirements by the condition~\eqref{eq:Ser2}.
\end{theorem}
\begin{proof}
The proof is to verify the non-decreasing property of the sum rate along with iterations until reaching a fixed point solution. The detailed proof is available in Appendix~\ref{Appendix:SelectedUser}.
\end{proof}
After adding user~$k_m^{t,\ast}$ to the system, we should update the available- and scheduled-user sets $\mathcal{N}$ and $\mathcal{K}(t)$ as 
\begin{equation} \label{eq:Updatev1}
	\mathcal{N}(t) \leftarrow \mathcal{N}(t) \setminus \{ k_m^{t,\ast} \} \mbox{ and } \mathcal{K}_m(t) \leftarrow \widetilde{\mathcal{K}}_{m}^{\ast} (t).
\end{equation}
The inner loop will continue until $m = M$ and the scheduled-user set $\mathcal{K}(t)$ is defined as
\begin{equation} \label{eq:Kt}
\mathcal{K}(t) \leftarrow \widetilde{\mathcal{K}}_M(t).
\end{equation}
At the end of each outer iteration, the algorithm should remove scheduled users from service if they are already satisfied their QoS requirements. This is done by computing the aggregated throughput in \eqref{eq:Rk}, and checking the QoS condition:
\begin{equation} \label{eq:Conditionv1}
R_k ( \{ \mathcal{K} (t)\}) \geq \xi_k.
\end{equation}
Let us denote $\widehat{\mathcal{K}}(t) \subseteq \mathcal{K}(t)$ the set of scheduled users already satisfied their QoS requirements, $\mathcal{K}(t)$ is further updated as
\begin{equation} \label{eq:Ktv1}
	\mathcal{K}(t) \leftarrow \mathcal{K}(t)\setminus \widetilde{\mathcal{K}}(t).
\end{equation}
The iterative approach will continue until all the time slots are considered and the proposed heuristic approach is summarized in Algorithm~\ref{Algorithm1}. Despite the local user scheduling solution, our proposed approach ensures the long-term sum throughput maximization over many different time slots with respect to their individual QoS requirements. Observing the solution of Algorithm~\ref{Algorithm1}, two fundamentals are: $i)$ the system always offers data throughput to the scheduled users at least equal to the individual QoS requirements; $ii)$ Some available users may not be scheduled in the observed window time due to mutual interference caused by sharing the time and frequency resource plane and the fixed power allocation.
 \begin{remark}
 For a given power level, Algoritm~\ref{Algorithm1} performs the user scheduling that maximizes the total sum throughput on a long-term period with many time slots, while strictly guarantees the individual QoS whenever a scheduled user is allowed to join the network. Aligned with previous works, one can attain a better sum throughput than solving \eqref{Problemv1} by relaxing the QoS constraints which enlarges the feasible domain. As a consequence, the relaxation  may result in some scheduled users served by lower than what they requested. However, we expect that the power control presented later will compensate the loss and all the users can be possibly served with their demands.
 \end{remark}
\vspace*{-0.25cm}
\subsection{Computational Complexity}
The computational complexity of Algorithm~\ref{Algorithm1} is now analytically presented. Let us consider the multiplications, division, square root, and matrix inversion as the dominated arithmetic operations, similar to \cite{van2019power,massivemimobook}, the computational complexity order of  Algorithm~\ref{Algorithm1} is given in Lemma~\ref{eq:Complexity}.
\begin{lemma} \label{eq:Complexity}
	Algorithm~\ref{Algorithm1} has the computational complexity in the order of $\mathcal{O}\left( C_0 + C_1 + C_2 \right)$, where $C_0, C_1, C_2$ are given as 
	\begin{align}
		C_0 &= NM + N\log_2 N,\\
		C_1 &= (M+2) \sum\limits_{t=1}^T \sum\limits_{m= |\mathcal{K}(t-1)|}^M  |\mathcal{N}(t)| |\widetilde{\mathcal{K}}_m(t)|,\\
		C_2 & = \frac{M^2}{2}\sum\limits_{t=1}^T \sum\limits_{m=|\mathcal{K}(t-1)|}^M |\mathcal{N}(t)|   |\widetilde{\mathcal{K}}_m(t)|^2.
	\end{align}
\end{lemma}
\begin{proof}
	Selecting the first scheduled user based on the channel gains requires the $N(M+1)$ arithmetic operations to compute the $N$ channel gains and $\mathcal{O}(N\log_2 N)$ for sorting them in a descending order as in \eqref{eq:ChannelOrder}. Therefore, the computational complexity of this step is proportional to $N( M+ 1 + \log_2 N )$. For each inner loop, we first need to compute the instantaneous throughput in \eqref{eq:Capacityk}, which requires the $(M+2)|\widetilde{\mathcal{K}}_m(t)| +3$ arithmetic operations. The computational complexity needed to solve \eqref{Prob:Ratev1} scales up with the factor $ |\mathcal{N}(t)|(M+2) |\widetilde{\mathcal{K}}_m(t)|+ 3 |\mathcal{N}(t)| $, thus the inner loop has the computational complexity in the order of $ |\mathcal{N}(t)|(M+2) \sum_{m= |\mathcal{K}(t-1)|}^M  |\widetilde{\mathcal{K}}_m(t)|$. Furthermore, each RZF precoding matrix with the cost as in Lemma~\ref{lemma:ComplexityWtv2} leads to the total computational complexity per inner loop in the order of $ \frac{1}{2}|\mathcal{N}(t)|M^2 \sum_{m=|\mathcal{K}(t-1)|}^M  |\widetilde{\mathcal{K}}_m(t)|^2$. By summing up all the cost and removing the terms with low degree, the result is obtained as in the lemma.
\end{proof}
Lemma~\ref{eq:Complexity} manifests that Algorithm~\ref{Algorithm1} has the computational complexity mainly spending on selecting the new scheduled user, updating the instantaneous channel capacity in \eqref{eq:Rk}, and recomputing the precoding matrix in \eqref{eq:MMSEPrecodMatrix} whenever a unscheduled new user is considered. Among those contributors, the precoding matrix recomputation consumes the highest cost, which is in a quadratic order of the scheduled users and satellite beams  per iteration. Nonetheless, the entire computational complexity is much lower than an exhaustive search. This algorithm can thus perform the user scheduling for a large-scale network with many users.
\vspace*{-0.25cm}
\section{Joint User Scheduling and Data Power Allocation for Sum Rate Maximization} \label{Sec:Joint}
As an extension, this section digs into the benefits of user scheduling and power allocation in boosting the sum throughput with the individual QoS constraints. In particular, we formulate and solve an optimization problem that maximizes the sum throughput subject to a limited power budget at the satellite via jointly considering both the scheduled-user set and power coefficients as optimization variables. 
\vspace*{-0.25cm}
\subsection{Problem Formulation}
By optimizing both the scheduled-user set and power coefficients to enhance the network performance, we would like to solve the following problem:
\begin{subequations} \label{ProblemJoint}
	\begin{alignat}{2}
		& \underset{\{ \mathcal{A}(t) \} }{\mathrm{maximize}}
		& & \, \sum\limits_{t=1}^T \sum\limits_{k \in \mathcal{K}(t) } R_k(\mathcal{A}(t)) \label{eq:Objv1} \\
		& \mbox{subject to}
		& &   \, R_k(\mathcal{A}(t)) \geq \xi_k /T_k, \forall k \in \mathcal{K}(t), \forall t, \label{eq:QoSconstJ}\\
		&&& \, \mathcal{K} (t) \subseteq \{ 1, \ldots, N\}, \forall t, \label{eq:Kt1J}\\
		&&& \, | \mathcal{K} (t)| \leq M, \forall t, \label{eq:Kt2J} \\
		&&& \,  \cup \mathcal{K}(t) \subseteq \{1, \ldots,N \} \label{eq:Kt3J}, \\
		&&& \sum\limits_{k' \in \mathcal{K}(t)} p_{k'}(t) \leq P_{\max}, \forall t,  \\
		&&& p_{k}(t) \geq 0, \forall k \in \mathcal{K}(t) , \forall t, \label{eq:Power1}\\
		&&& \mathcal{A}(t) = \mathcal{K}(t) \cup \{ p_k (t) \}, \forall t. \label{eq:Power2}
	\end{alignat}
\end{subequations}
In comparison to the previous problem in \eqref{Problemv1}, together with the scheduled-user set $\mathcal{K}(t)$, the power coefficients are new optimization variables added to \eqref{ProblemJoint} in each time slot. The practical power constraints in \eqref{eq:Power1} and \eqref{eq:Power2} introduce extra complexity to find the optimal solution. In addition, the objective function \eqref{eq:Objv1} and the QoS constraints \eqref{eq:QoSconstJ} are  challenging to manipulate as they are  multivariate functions of both the scheduled users and power coefficients. The feasible domain of problem~\eqref{ProblemJoint} includes the hybrid optimization variables, either continuous as the transmit power coefficients or discrete as the scheduled-user set. In addition, we stress that the combinatorial structure of \eqref{ProblemJoint} retains from \eqref{Problemv1}; thus, we extend Algorithm~\ref{Algorithm1} to design the joint user scheduling and power allocation solution that performs resource allocation for many available users.
\vspace*{-0.25cm}
\subsection{Two-Stage Algorithm}
For a low computational complexity design, we tackle the non-convexity of problem~\eqref{ProblemJoint} and find a local optimum by proposing a two-stage approach as shown in Algorithm~\ref{Algorithm2}. The first stage focuses on the user scheduling with the fixed transmit power coefficients by the main steps presented in Section~\ref{Sec:2StageAlgo}. The fixed power coefficients may not maximize the strength of desired signals and effectively mitigate mutual interference, but it poses a basic trends for the user scheduling. Hence, in Stage~$1$, we can either implement one of the following options:
\begin{itemize}
 \item[$i)$] Keep the monotonic property of the sum throughput and the individual QoS requirements, as  \eqref{eq:Ser1} and \eqref{eq:Ser2} in the per-time-slot selection of scheduled users. The proposed algorithm will guarantee all the scheduled users always satisfying their individual QoS requirements along with iterations. By following this option, some users may be in unfeasibilities, but the per-user throughput of scheduled users can be  significantly improved as confirmed by numerical results later.
  \item[$ii)$] Relax the conditions \eqref{eq:Ser1} and \eqref{eq:Ser2} in the per-time-slot selection of scheduled users  by supposing that the power allocation will enhance the QoSs of scheduled users with slightly lower data throughput than requested. We must assume that all scheduled users should satisfy their QoSs with the power allocation. By following this option, the sum throughput is maximized.
\end{itemize}
The above options will determinine the scheduled-user set $\mathcal{K}(t)$ for a fixed power level in Stage~$1$ as shown in Step~3.1 in Algorithm~\ref{Algorithm2}. Then, Stage~$2$ reformulates  problem~\eqref{ProblemJoint} so that the transmit power coefficents of all the scheduled users in $\mathcal{K}(t)$ are reallocated at the $t$-th time slot as
\begin{equation} \label{ProblemvData}
	\begin{aligned}
		& \underset{\{ p(t) \} }{\mathrm{maximize}}
		& & \,  \sum\limits_{k \in \mathcal{K}(t) } R_k( \{ p(t) \} ) \\
		& \mbox{subject to}
		& &   \, R_k(\{ p(t) \} ) \geq \frac{\xi_k}{T_k}, \forall k \in \mathcal{K}(t),\\
		&&& \sum\limits_{k' \in \mathcal{K}(t)} p_{k'}(t) \leq P_{\max}, \\
		&&& p_{k}(t) \geq 0, \forall k \in \mathcal{K}(t).
	\end{aligned}
\end{equation}
Compared to \eqref{ProblemJoint}, problem~\eqref{ProblemvData} simplifies the matter since its feasible set is convex,  and the instantaneous throughput only depends on the data power coefficients. This problem is still non-convex due to the objective function. According to Weierstrass' theorem \cite{van2018joint,Horn2013a}, an optimal solution always exists, but  the individual QoS constraints make this problem nontrivial to obtain the global optimum. By using the instantaneous throughput \eqref{eq:Rk} in Lemma~\ref{lemma:ChannelCapacity} together with the  monotonic property of the sum of the logarithm functions, we can
reformulate problem~\eqref{ProblemvData} to an equivalent form as
\begin{equation} \label{Prob:SubProbv3}
	\begin{aligned}
		& \underset{\{ p_{k} (t) \} }{\mathrm{maximize}}
		&&  \prod\limits_{k  \in \mathcal{K}(t) }  \left( 1 + \mathrm{SINR}_{k}  \left( \{ p_{k}(t) \}\right)  \right)  \\
		& \,\,\mathrm{subject \,to}
		&& \mathrm{SINR}_{k}  \left( \{ p_{k}(t) \} \right) \geq \nu_k , \forall k \in \mathcal{K}(t), \\
		& & &   \sum\limits_{k \in \mathcal{K}(t) } p_{k}(t) \leq P_{\max},\\
		&&& p_{k}(t) \geq 0, \forall k \in \mathcal{K}(t),
	\end{aligned}
\end{equation}
where $\nu_k = 2^{\xi_k/T_k} - 1$ indicates that the instantaneous throughput constraint is converted to the corresponding SINR constraint. By introducing the auxiliary variables $\gamma_{k}(t)$ for each scheduled user, we further represent problem~\eqref{Prob:SubProbv3} in an epigraph form \cite[pp. 134]{Boyd2004a}  as follows:
\begin{subequations} \label{Prob:SubProbv4}
	\begin{alignat}{2}
		& \underset{\{ p_{k}(t), \gamma_{k}(t)  \} }{\mathrm{maximize}}
		&&  \prod\limits_{k \in \mathcal{K}(t) }  \gamma_{k}(t) \label{eq:Obv1} \\
		& \mathrm{subject \,to}
		&&\quad  1 + \mathrm{SINR}_{k}  \left( \{ p_{k} (t) \}\right)  \geq \gamma_{k}(t), \forall k \in \mathcal{K}(t), \label{eq:Ratev1} \\
		&&& \quad \mathrm{SINR}_{k}  \left( \{ p_{k} (t) \}\right) \geq \nu_k , \forall k \in \mathcal{K}(t), \label{eq:QoSv1}\\
		&& & \quad \sum\nolimits_{k \in \mathcal{K}(t)} p_{k}(t) \leq P_{\max}, \label{eq:PowerBudget1} \\
		&&& p_{k}(t) \geq 0, \forall k \in \mathcal{K}(t).
	\end{alignat}
\end{subequations}
In fact, the auxiliary variables transform the objective function into a monomial and shift the main nonlinear parts into the constraints. Even though the objective function and the power budget constraint of problem~\eqref{Prob:SubProbv4} are convex, the main challenge comes from the SINR constraints. The standard form of problem~\eqref{Prob:SubProbv4} is stated in Lemma~\ref{lemma:Signomial}.
\begin{lemma} \label{lemma:Signomial}
Problem~\eqref{Prob:SubProbv4} is a signomial program,  which is non-convex on the standard form.
\end{lemma}
\begin{proof}
The objective function \eqref{eq:Obv1} of problem~\eqref{Prob:SubProbv4} is monomial and the limited power budget constraint \eqref{eq:PowerBudget1} is posynomial, they are convex (please see Appendix~\ref{Appendix:UsefulDef} for more detail). The individual QoS constraint \eqref{eq:QoSv1} is reformulated as
\begin{equation}
\nu_k \sum\limits_{k' \in \mathcal{K}(t) \setminus \{ k\} } \frac{p_{k'} (t) \big|\mathbf{h}_{k}^H \mathbf{w}_{k'}(t)|^2}{p_{k} (t) \big| \mathbf{h}_{k}^H \mathbf{w}_{k}(t) \big|^2 } + \frac{\nu_k \sigma^2}{p_{k} (t) \big| \mathbf{h}_{k}^H \mathbf{w}_{k}(t) \big|^2 } \leq 1,
\end{equation}
which is also posynomial. The main proof is now to show \eqref{eq:Ratev1} is signomial that is nonconvex. In particular, we recast \eqref{eq:Ratev1} to an equivalent form as
\begin{equation} \label{eq:FirstConstraint}
	\begin{split}
		& \gamma_{k}(t)\sum\limits_{k' \in \mathcal{K}(t) \setminus \{ k\}} p_{k'}(t) \left| \mathbf{h}_{k}^H \mathbf{w}_{k'}(t) \right|^2 + (\gamma_{k}(t) -1 ) \sigma^2  \\
		& - \sum\limits_{k' \in \mathcal{K}(t) } p_{k'}(t) \left| \mathbf{h}_{k}^H \mathbf{w}_{k'}(t) \right|^2  \leq 0,
	\end{split}
\end{equation}
which is a signomial constraint since the left-hand side of \eqref{eq:FirstConstraint} is a signomial function (please see Definition~\ref{Def1} in Appendix~\ref{Appendix:UsefulDef}). Consequently, problem~\eqref{Prob:SubProbv4} is a signomial program as stated in the lemma, concluding the proof.
\end{proof}
Apart from pointing out the inherent non-convexity structure of problem~\eqref{Prob:SubProbv4}, Lemma~\ref{lemma:Signomial} gives a clue to obtain a local solution  by exploiting the signal programming features \cite{Chiang2007b}. The successive optimization approach is deployed to tackle problem~\eqref{Prob:SubProbv4} by following a similar methodology as in \cite{van2018joint}. We now introduce the arithmetic mean-geometric mean inequality to cope with the signomial SINR constraints of problem~\eqref{Prob:SubProbv4}. Specifically, each signomial SINR constraint in every group should be converted to a posynomial one by utilizing the below approximation.
\begin{lemma}\cite[Lemma~$1$]{Chiang2007b} \label{lemma:AMGM}
	Let us consider a posynomial function $f(x)$ defined as a summation from  a set of $|\mathcal{K}(t)|$ monomial functions $\{ h_1(x), \ldots, h_{|\mathcal{K}(t)|}(x) \}$ as
	\begin{equation}
		f(x) = \sum_{k=1}^{|\mathcal{K}(t)|} h_k (x),
	\end{equation}
	then $f(x)$ is lower bounded by a monomial function $\tilde{f}(x)$, which is defined as
	\begin{equation}
		f(x) \geq \tilde{f}(x) = \prod\limits_{k \in \mathcal{K}(t) } \left(\frac{ f_k(x)}{\mu_k} \right)^{\mu_k},
	\end{equation}
	where $\mu_k$ is a non-negative weight of the monomial function $f_k(x)$. Notice that $\tilde{f}(x_0)$ is the best approximation near $x_0$ by the first-order Taylor expansion when the weight $\mu_k$  assigned to scheduled user~$k$ is computed as
	\begin{equation} \label{eq:Weight}
		\mu_k = \frac{f_k(x_0)}{\sum\limits_{k'=1}^{|\mathcal{K}(t)|} f_{k'}(x_0)}.
	\end{equation}
\end{lemma}
Through Lemma~\ref{lemma:AMGM}, the non-convexity of problem~\eqref{Prob:SubProbv4} is manipulated by introducing the function $\tilde{g}_{k}(\{ p_{k'}(t) \} )$ associated with scheduled user~$k$ as 
\begin{equation} \label{eq:gkm}
	\tilde{g}_{k}(\{ p_{k'}(t) \} ) = \left(\frac{\sigma^2}{\mu_{0k}(t)} \right)^{\mu_{0k}(t)}\prod_{k'\in \mathcal{K}(t)} \left( \frac{p_{k'}(t)z_{kk'}(t)}{\mu_{kk'}(t)} \right)^{\mu_{kk'}(t)} ,
\end{equation}
where $z_{kk'}(t) =  | \mathbf{h}_{k}^H \mathbf{w}_{k'}(t) |^2$ and the non-negative weights, $\mu_{0k}(t), \mu_{kk'}(t),$ satisfy the normalization condition:
\begin{equation} \label{eq:WeightsSINRv1}
	\mu_{0k}(t) + \sum\limits_{k'\in \mathcal{K}(t)} \mu_{kk'}(t) = 1.
\end{equation} 
After that, a stationary solution to problem~\eqref{Prob:SubProbv4} can be obtained in polynomial time as stated in Lemma~\ref{theorem:geometricprog}.
\begin{lemma} \label{theorem:geometricprog}
The global optimum to problem~\eqref{Prob:SubProbv4} is lower bounded by the solution of the following geometric program
	\begin{equation} \label{Prob:SubProbv5}
		\begin{aligned}
			& \underset{\{ p_{k} (t) \geq 0, \gamma_{k}(t) \} }{\mathrm{maximize}}
			&&   \prod\limits_{k \in \mathcal{K}(t) }   \gamma_{k}(t)  \\
			& \,\,\mathrm{subject \,to}
			&& \mathrm{Constraint} \, \, \eqref{eq:SINRConstraintv1}, \forall k \in \mathcal{K}(t),\\
			&&& \mathrm{SINR}_{k}  \left( \{ p_{k} (t) \}\right) \geq \nu_k , \forall k \in \mathcal{K}(t), \\
			&& &   \sum\limits_{k \in \mathcal{K}(t)} p_{k}(t) \leq P_{\max},
		\end{aligned}
	\end{equation}
	where the constraint \eqref{eq:Ratev1} for scheduled user~$k$ is lower bounded by the following constraint:
	\begin{equation} \label{eq:SINRConstraintv1}
		\begin{split}
		\gamma_k(t) \sum\limits_{k'\in \mathcal{K}(t)\setminus \{ k\}}  \frac{p_{k'}(t) z_{kk'}(t)}{\tilde{g}_{k}(\{ p_{k'}(t) \} )} + \frac{\gamma_k(t)\sigma^2}{\tilde{g}_{k}(\{ p_{k'}(t) \} )} \leq 1.
		\end{split}
	\end{equation}
\end{lemma}
\begin{proof}
	The proof is to show the SINR constraints of problem~\eqref{Prob:SubProbv4} are bounded by the corresponding expressions in \eqref{Prob:SubProbv5} by utilizing the arithmetic mean-geometric mean inequality in Lemma~\ref{lemma:AMGM}. The detailed proof is available in Appendix~\ref{Appendix:geometricprog}.
\end{proof}
We stress that Lemma~\ref{theorem:geometricprog} helps us obtain a local optimum to problem~\eqref{Prob:SubProbv4} by solving the geometric program \eqref{Prob:SubProbv5}. Since \eqref{Prob:SubProbv5} involves a hidden convex structure, we can exploit the main steps reported in \cite{Boyd2004a} to obtain the global optimum. Fortunately, problem~\eqref{Prob:SubProbv5} is in standard form, which can be solved by utilizing the interior point methods from a general-purpose toolbox such as CVX \cite{cvx2015}.  To make the local solution better, we now exploit the successive optimization approach in an iterative manner. In more detail, from an initial set of the transmit powers $\big\{  p_{k}^{(0)}(t) \big\}$ in the feasible domain with $ p_{k}^{(0)}(t)  = P_{\max}/ |\mathcal{K}(t)|$, the weights are updated at the $i$-th iteration as
\begin{align}
	\mu_{0k}^{(i)}(t) &=  \frac{\sigma^2}{\sum\limits_{k'' \in \mathcal{K}(t)} p_{k''}^{(i-1)}(t) z_{kk''}(t) + \sigma^2 }, \label{eq:WeightsSINRO1} \\
	\mu_{kk'}^{(i)}(t) &= \frac{p_{k'}^{(i-1)}(t) z_{kk'}(t)}{\sum\limits_{k'' \in \mathcal{K}(t)} p_{k''}^{(i-1)}(t) z_{kk''}(t) + \sigma^2 }, \label{eq:WeightsSINRO2}
\end{align}
It should be noticed that the updates in \eqref{eq:WeightsSINRO1} and \eqref{eq:WeightsSINRO2} ensure the condition \eqref{eq:WeightsSINRv1}.
Subsequently, the global optimum to the geometric program \eqref{Prob:SubProbv5} is obtained, say $\big\{ \rho_{km}^{(i)} \big\}$. This iterative process will be terminated when the variation between two consecutive iterations is sufficiently small, for example,
\begin{equation}
	\left|\sum\limits_{k \in \mathcal{K}(t)} R_{k} \big( \{ p_{k'}^{(i-1)} (t)\} \big) -  R_{k} \big( \{ p_{k'}^{(i)} (t) \} \big) \right| \leq \epsilon, 
\end{equation}
where $\epsilon$ is a given accuracy. The proposed power allocation to maximize the total channel capacity in the entire system is summarized in Algorithm~\ref{Algorithm2}. 
\begin{algorithm}[t]
	\caption{A local solution to problem~\eqref{ProblemJoint} by successive optimization approach} \label{Algorithm2}
	\textbf{Input}:  Available-user set $\mathcal{N}(0) \leftarrow \{1, \ldots, N \}$; Scheduled-user set $\mathcal{K}(0) \leftarrow \emptyset$; Propagation channel vectors $\{\mathbf{h}_1, \ldots, \mathbf{h}_N \}$; QoS requirements $\{ \xi_1, \ldots, \xi_k \}$; Number of time slots $T$ and individual scheduled time slots $\{T_1, \ldots, T_N\}$; Transmit data powers $\{ p_1, \ldots, p_{N} \}$.
	\begin{itemize}
		\item[1.] Select scheduled user~$\pi_1$ based on the best channel gain as obtained in \eqref{eq:ChannelOrder}.
		\item[2.] Set $t=1$, then update $\mathcal{N}(1)$ and $\mathcal{K}(1)$ by \eqref{eq:1stUpdate}.
		\item[3.] \textbf{while} $t \leq T$ \textbf{do}
		\begin{itemize}
		\item[] \textit{\textcolor{blue}{Stage~$1$: Perform user scheduling}}
		\item[3.1.] Select the scheduled-user set   $\mathcal{K}(t)$ by applying Algorithm~\ref{Algorithm1} with or without checking  conditions \eqref{eq:Ser1} and \eqref{eq:Ser2}.
		\item[] \textit{\textcolor{blue}{Stage~$2$: Perform power allocation}}
		\item[3.2.] Set $p_k^{(0)} = p_k, \forall k \in \mathcal{K}(t),$ and compute the weight values $\{ \mu_{0k}^{(1)}, \mu_{kk'}^{(1)} \}$ by \eqref{eq:WeightsSINRO1} and \eqref{eq:WeightsSINRO2}. Set $i=1$.
		\item[3.3.] Iteration $i$:
		\begin{itemize}
			\item [3.3.1.] Solve problem~\eqref{Prob:SubProbv5} with the weight values $\{ \mu_{0k}^{(i)}, \mu_{kk'}^{(i)} \}$ to get the optimal powers $\{ p_{k}^{(i)} (t) \}$.
			\item [3.3.2.] Update the weight values  $\{ \mu_{0k}^{(i+1)}, \mu_{kk'}^{(i+1)} \}$ from $\{ p_{k}^{(i)} (t) \}$ by \eqref{eq:WeightsSINRO1} and \eqref{eq:WeightsSINRO2}.
		\end{itemize}
	    \item[3.4.] Set $t=t+1$ and repeat Steps $3.3.1$ and $3.3.2$ until convergence.
		\end{itemize}
	    \item[4.] \textbf{end while}
	\end{itemize}
	\textbf{Output}: Scheduled user sets and transmit power coefficients. 
	 \vspace*{-0.0cm}
\end{algorithm}
At the convergence, a fixed point solution has the property as shown in Theorem~\ref{theorem:KKT}.
\begin{theorem} \label{theorem:KKT}
	The power solution obtained by Algorithm~\ref{Algorithm2} at the $i$-th time slot converges to a fixed point, which is a Karush–Kuhn–Tucker (KKT) point of problem~\eqref{Prob:SubProbv4}.
\end{theorem}
\begin{proof}
	The proof is adapted from the general framework in \cite{Marques1978a} to our system model and notation, whose main steps are sketched in Appendix~\ref{Appendix:KKT}.
\end{proof}

For the power allocation at the $i$-th time slot, problem~\eqref{Prob:SubProbv5} includes $2 |\mathcal{K}(t)|$ optimization variables and $2 |\mathcal{K}(t)| +1$ constraints, thus the computational complexity to solve this problem by the interior-point methods is of the order of
\begin{equation}
C_4 (t) = \mathcal{O} \big( L (t) L^{(i)}(t)  \max\{ 8 |\mathcal{K}(t)|^3 + 4 |\mathcal{K}(t)|^2, F^{(i)}(t)\}\big),
\end{equation}
where
 $F^{(i)}(t)$ is the cost of computing the first and second derivatives of the objective and constraints functions of problem~\eqref{Prob:SubProbv5} by utilizing the set of weight values $\big\{ \mu_{0k}^{(i)}(t),  \mu_{kk'}^{(i)}(t) \big\}$ as defined in \eqref{eq:WeightsSINRO1}-\eqref{eq:WeightsSINRO2};   
 $L^{(i)}(t)$ is the number of iterations that the interior-point methods need to obtain the solution to problem~\eqref{Prob:SubProbv5}. As reported in \cite{Boyd2004a}, $L^{(i)}(t)$ is in the range between $10$ and $100$; and $ L (t)$ is the number of iterations required to obtain a KKT point solution as stated in Theorem~\ref{theorem:KKT}. Therefore, the power allocation spends for all the observed window time $T$ is in the order of $C_4 = \sum_{t=1}^T C_4(t)$. In a nutshell, we recall that the computational complexity of Algorithm~\ref{Algorithm2} is the total cost of the user scheduling and the power allocation, which is roughly $\mathcal{O}(C_1 + C_2 + C_3 + C_4)$. 
 \begin{remark}
 For a given user-scheduling set, the proposed power allocation approach is analytically proved to reach a KKT point solution to problem~\eqref{ProblemvData} after a number of iterations. Nonetheless, the power solution is a local optimum of the original problem~\eqref{ProblemJoint} due to its nonconvex structure. Algorithm~\ref{Algorithm2} is expected to provide a preliminary toolbox to evaluate the sum throughput over a window time while maintaining the individual QoS requirements of scheduled users under the limited transmit power at the GEO satellite. Proper user scheduling makes it reasonable for the considered framework to assume that the system should serve all the available users with equal or higher than their QoS requirements. Nevertheless, the congestion issue might appear in practice when at least one user cannot satisfy its demand due to the limited power budget \cite{bui2021robust}. Subsequently, user scheduling and congestion control should be an interesting extension.  
 \end{remark}
 \begin{figure}[t]
 	\centering
 	\includegraphics[trim=0.6cm 0.0cm 1.2cm 0.8cm, clip=true, width=3.5in]{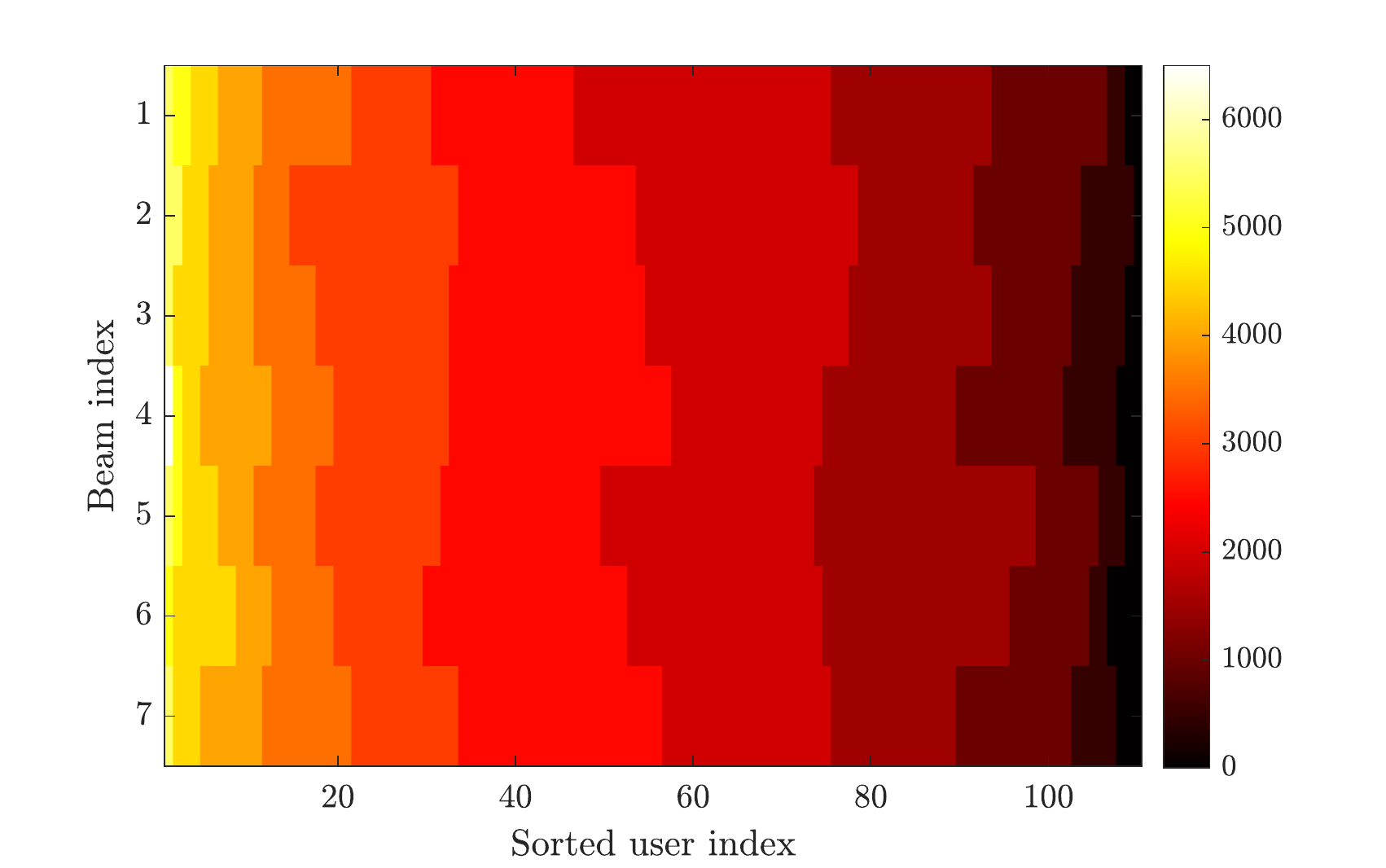} \vspace*{-0.25cm}
 	\caption{The individual QoS requirement [Mbps] per user in each beam over the observed window time.}
 	\label{FigQoSRequirement}
 \end{figure}
\begin{figure}[t]
	\centering
	\includegraphics[trim=0.6cm 0.0cm 1.2cm 0.6cm, clip=true, width=3.5in]{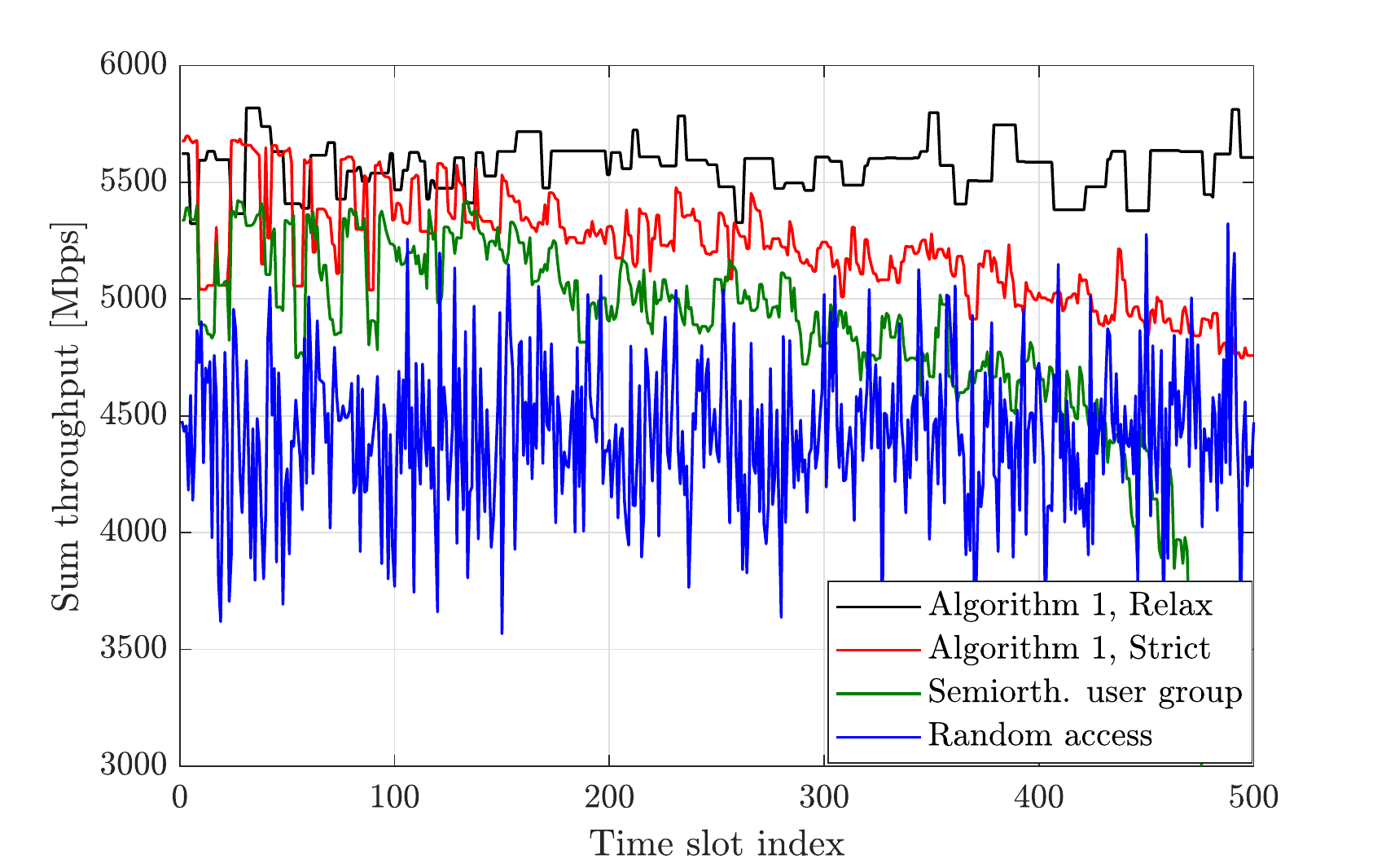} \vspace*{-0.25cm}
	\caption{The sum throughput [Mbps] over the observed window time with a fixed transmit power level.}
	\label{FigFixedPowerSumRate}
\end{figure}
\begin{figure}[t]
	\centering
	\includegraphics[trim=0.6cm 0.0cm 1.2cm 0.6cm, clip=true, width=3.5in]{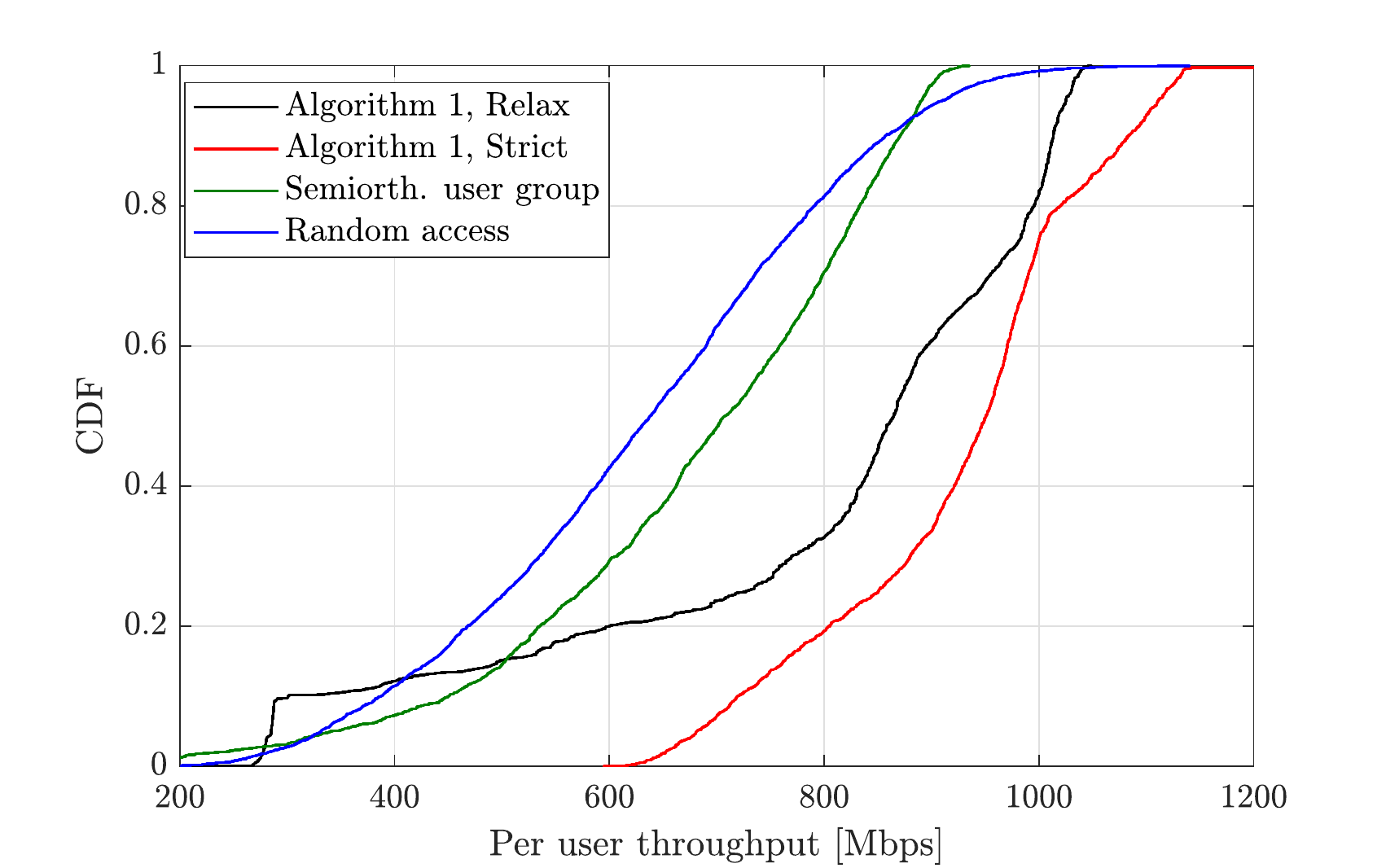} \vspace*{-0.25cm}
	\caption{The CDF of the per user throughput [Mbps] over the observed window time with a fixed transmit power level.}
	\label{FigResultedPerUserRate}
	\vspace*{-0.2cm}
\end{figure}
\begin{figure}[t]
	\centering
	\includegraphics[trim=0.6cm 0.0cm 1.2cm 0.6cm, clip=true, width=3.5in]{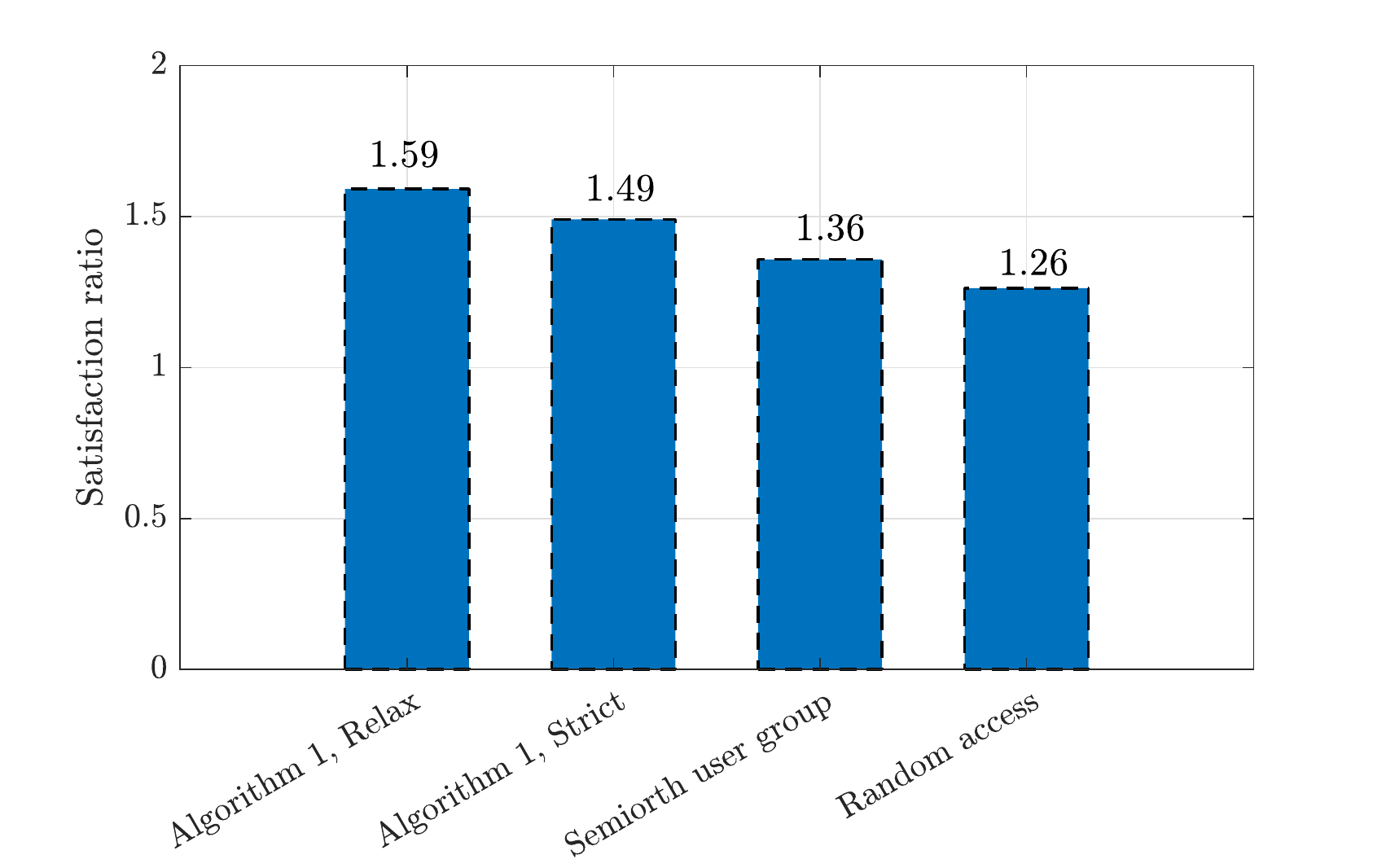} \vspace*{-0.25cm}
	\caption{The satisfaction ratio between the aggregated throughput [Mbps] in \eqref{eq:Rk} and the individual QoS requirement over the observed window time with a fixed transmit power level.}
	\label{FigFixSatisRatio}
	\vspace*{-0.2cm}
\end{figure}
\vspace*{-0.2cm}
\section{Numerical Results} \label{Sec:NumericalResults}
We consider a GEO satellite system with seven beams each serving $110$ users, uniformly distributed in each beam. The observed window time includes $500$ time slots. A sum power-constrained system is regarded with the per-beam power of $10$~dBW. The system bandwidth is $500$~MHz, and the carrier frequency is $19.95$~GHz. The maximum transmit power is $P_{\max} = 18.45$~[dbW] that is related to the average beam-pattern gain $44.4$~[dBi] and the effective isotropic radiated power (EIRP) $-27$~[dbW/Hz]. The noise variance is $-118.3$~[dB], corresponding to the noise figure $2.28$~dB. The receive antenna diameter is $0.6$~m with an efficiency $0.6$. For simplicity, we normalize that each time slot to one second and the average QoS requirement per time slot is $500$~Mbps. The total number of time slots occupied by per scheduled user, $T_k , \forall k,$ are in the range $[0, 13]$ by a uniform distribution. Consequently, the individual QoS requirements for the available users in all the beams are shown in Fig.~\ref{FigQoSRequirement}. We sort the users' QoS requirements in a descending order varying from $0$~Mb to $6500$~Mbps over the observed window time for better visualization. The system performance is evaluated with either a fixed or optimized power allocation. The following benchmarks are included for comparison  to demonstrate the efficiency of the proposed optimization frameworks:
\begin{itemize}
\item[$i)$] \textit{Proposed user scheduling algorithm} is presented in Algorithm~\ref{Algorithm1} with a fixed power level that guarantees the individual QoS constraints (Algorithm~\ref{Algorithm1}, Strict) and without the QoS constraints (Algorithm~\ref{Algorithm1}, Relax). 
\item[$ii)$] \textit{Joint user scheduling and power allocation algorithm} is presented in Algorithm~\ref{Algorithm2} that takes the double benefits by selecting good users and performing power allocation to maximize the sum throughput. In Stage~$1$, if the user scheduling algorithm does not guarantee the individual QoSs, we denote the solver as  Algorithm~\ref{Algorithm2}, Relax. Otherwise, it is designated as Algorithm~\ref{Algorithm2}, Strict. 
\item[$iii)$] \textit{Semiorthogonal user group} was proposed in \cite{Yoo2006a} by exploiting the orthogonality among  propagation channels at time-slot level. The number of scheduled users and satellite beams are assumed to be equal. Additionally, the user scheduling does not include the QoS requirements into account.
\item[$iv)$] \textit{Random access} is a low computational complexity benchmark and served as the baseline in previous works \cite{8371220}. Along with time slots, the number of scheduled users is randomly selected and equal to the number of satellite beams. There is no guarantee on the QoS requirements.
\end{itemize}
\begin{figure}[t]
	\centering
	\includegraphics[trim=0.6cm 0.0cm 1.2cm 0.6cm, clip=true, width=3.5in]{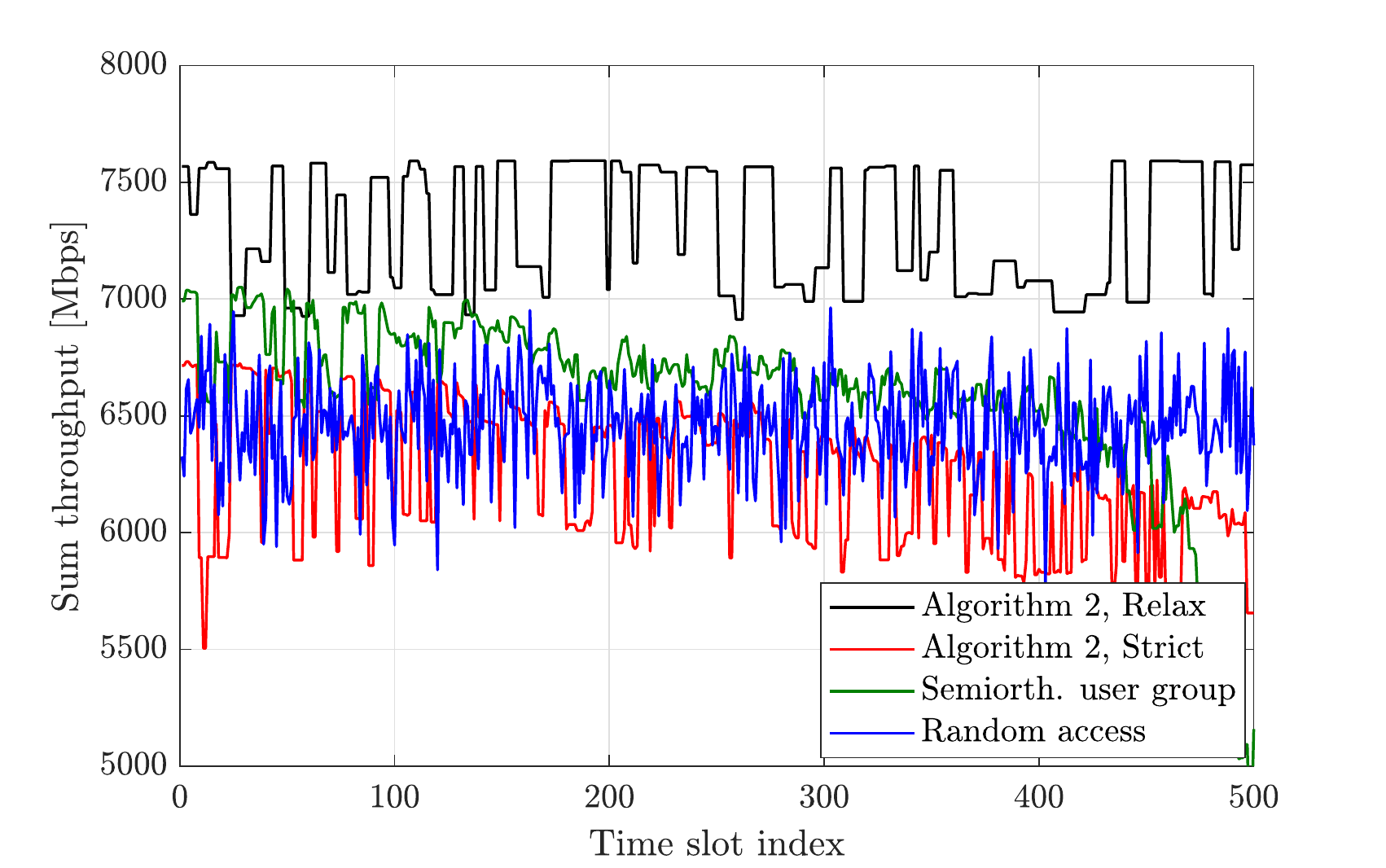} \vspace*{-0.25cm}
	\caption{The sum throughput [Mbps] over the observed window time with the power allocation.}
	\label{FigControlPowerSumRate}
	\vspace*{-0.2cm}
\end{figure}
\begin{figure}[t]
	\centering
	\includegraphics[trim=0.6cm 0.0cm 1.2cm 0.6cm, clip=true, width=3.5in]{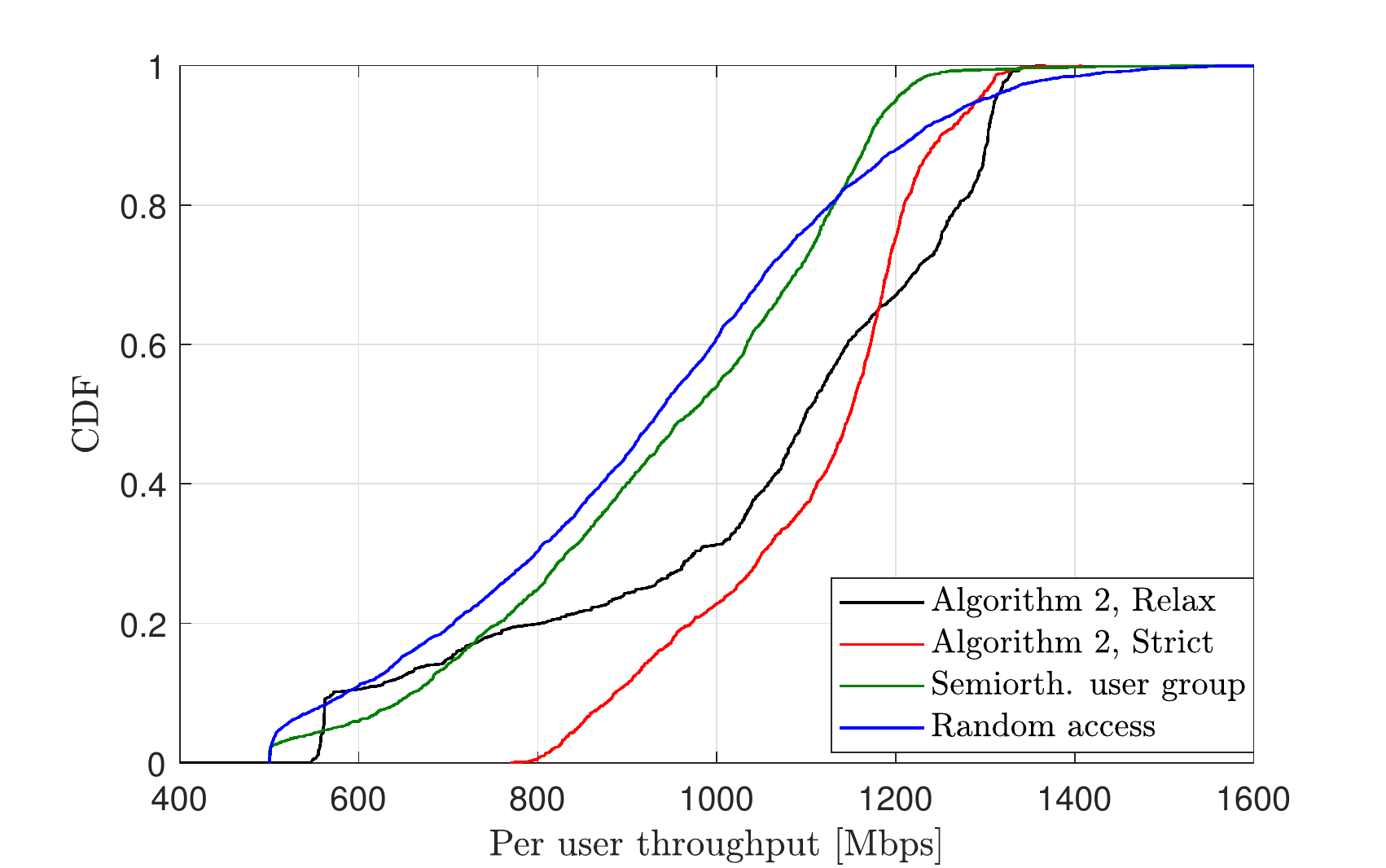} \vspace*{-0.25cm}
	\caption{The CDF of the per user throughput [Mbps] over the observed window time with the power allocation.}
	\label{FigResultedPerUserRateControl}
	\vspace*{-0.2cm}
\end{figure}
\begin{figure}[t]
	\centering
	\includegraphics[trim=0.6cm 0.0cm 1.2cm 0.6cm, clip=true, width=3.5in]{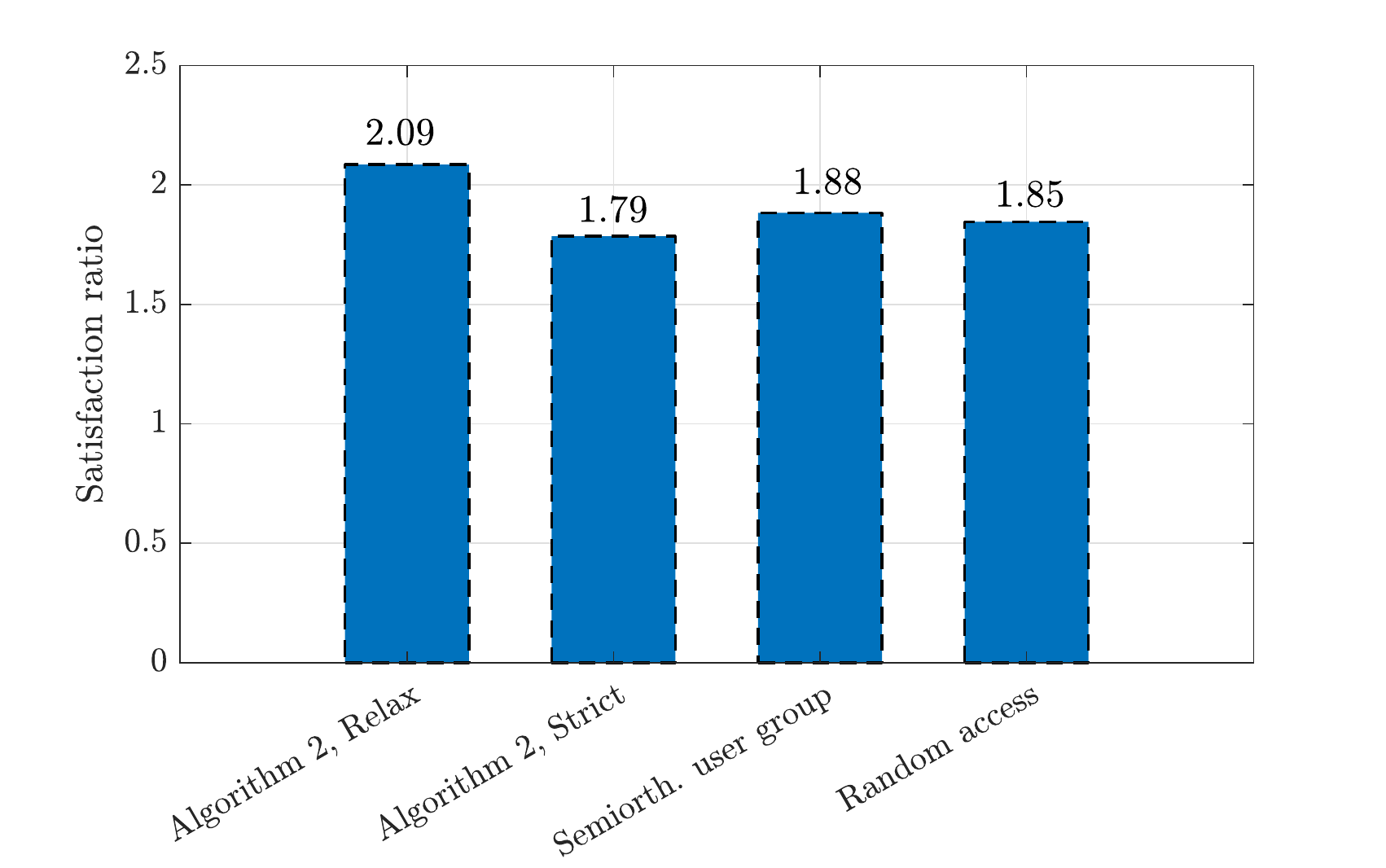} \vspace*{-0.25cm}
	\caption{The satisfaction ratio between the aggregated throughput [Mbps] in \eqref{eq:Rk} and the individual QoS requirement over the observed window time with power allocation.}
	\label{FigControlSatisRatio}
	\vspace*{-0.2cm}
\end{figure}
\begin{table}[t]
	\begin{center}
		\caption{Gain by Power Allocation.}
		\label{tab:table1}
		\begin{tabular}{|c|c|c|c|c|} 
			\hline 
			Benchmark & Proposed, & Proposed, & Semiorth. & Random\\
			& Relax & Strict & user group & access \\
			\hline
			Sum throughput& $1.31\times$ & $1.20\times$ &$1.44\times$ & $1.47\times$\\
			\hline
			Per user throughput& $1.37\times$ & $1.20\times$ &$1.39\times$ & $1.46\times$\\
			\hline
			Satisfaction ratio & $1.31\times$  & $1.20 \times $& $1.38\times$& $1.47 \times$ \\
			\hline
		\end{tabular}
	\end{center}
   \vspace*{-0.2cm}
\end{table}
In Fig.~\ref{FigFixedPowerSumRate}, we plot the sum throughput [Mbps] as a function of the time slot index. Random access provides the worst sum throughput in most of the time slots, only $4419$ [Mbps] on average. However, it offers good performance in the last time slots, which is better than semiorthogonal user group and Algorithm~\ref{Algorithm1} with the strict requirement on the QoSs. Semiorthogonal user group outperforms random access by $7.6\%$ of the sum throughput with $4753$ [Mbps] on average. It becomes the worst in the last time slots where the available users have strongly correlated channels. Algorithm~\ref{Algorithm1}, with the strict requirements on the individual QoSs,  performs $18.6\%$ better than the baseline. By relaxing the QoS constraints, Algorithm~\ref{Algorithm1} gives the best sum throughput with the gain of $6.7\%$ compared to the second-best benchmark, i.e. Algorithm~\ref{Algorithm1}, Strict.

In Fig.~\ref{FigResultedPerUserRate}, we show the cumulative density function (CDF) of the scheduled users over the observed window time. Random access averagely provides the throughput of about $631$ [Mbps] per user, while semiorthogonal user group offers $679$ [Mbps]. Notably, Algorithm~\ref{Algorithm1} gives the highest per-user throughput with $1.37\times$ higher than semiorthogonal user group. Algorithm~\ref{Algorithm1} ensures all the scheduled users with their QoS requirements. In contrast,  $24.4\%$ and $14.5\%$ of user locations cannot be served with the requested QoS if the system deploys random access and semiorthogonal user group, respectively, due to no QoS guarantee in those benchmarks. It manifests the practical importance of Algorithm~\ref{Algorithm1}. To guarantee the individual QoS requirements, Algorithm~\ref{Algorithm1}, Strict, has ignored some users from service resulting in the highest per-user throughput for all the scheduled users.

The satisfaction ratio is computed as a faction between the served data throughput and the required QoS over the observed window time, which is numerically shown in Fig.~\ref{FigFixSatisRatio}. All the satisfaction ratios are more significant than one meaning that the satellite system can provide data throughput above the individual QoS requirements on average, even with a fixed power level. Even though some scheduled users may not be satisfied with their service in a few time slots, Algorithm~\ref{Algorithm1}, with the relaxed constraints, offers the good satisfaction ratio for an extended period. The following are Algorithm~\ref{Algorithm1} with the strict individual QoS requirements (Algorithm~\ref{Algorithm1}, Strict), semiorthogonal user group, and random access. Despite its simplicity, random access provides a pretty good satisfaction ratio,  $26.19\%$ lower than Algorithm~\ref{Algorithm1}, Relax.

We now testify to the contributions of power allocation. Fig.~\ref{FigControlPowerSumRate} plots the sum throughput [Mbps] for all the considered benchmarks over the observed window time. It shows up superior improvements of Algorithm~\ref{Algorithm2} with the relaxation on Stage~$1$. It verifies that many scheduled users with lower data throughput than their demands can be improved by optimizing the transmit power. On average, Algorithm~\ref{Algorithm2} with relaxation offers the sum throughput $7300$~[Mbps]. Following, semiorthogonal user group averagely offers the sum throughput $6593$~[Mbps]. Under the power allocation, this benchmark works well at the beginning, where many semiorthogonal channels are available. Algorithm~\ref{Algorithm2} with the strict requirements Stage~$1$ ignores  many scheduled users, so the sum throughput is only about $6250$~[Mbps]. Random access gives the average sum throughput is about $6461$~[Mbps].

In Fig.~\ref{FigResultedPerUserRateControl}, we display the CDF of the per scheduled user throughput [Mbps]. We observe that the system can provide the required QoSs to all the scheduled users thanks to the power allocation in Algorithm~\ref{Algorithm2}. Random access gives the average per user throughput is about $923$~[Mbps], while semiorthogonal user group is about $942$~[Mbps], $2.06\%$ improved than the baseline. Thanks to the power allocation, the gap between semiorthogonal user group and random access is shortened. Our proposed algorithms give the best performance of $1403$~[Mbps] and of $893$~[Mbps] for Algorithm~\ref{Algorithm2}, Strict, and Algorithm~\ref{Algorithm2}, Relax, respectively. 

The satisfaction ratio for the system with power allocation is shown in Fig.~\ref{FigControlSatisRatio}. On average, Algorithm~\ref{Algorithm2}, with the relaxation on Stage~$1$ (Algorithm~\ref{Algorithm2}, Relax), offers the best satisfaction level. In contrast, Algorithm~\ref{Algorithm2}, with the strict individual QoS requirements on Stage~$1$, yields the lowest satisfaction level with an improvement of only $1.79 \times$. Many users whose data throughput could have been improved in Stage~$2$ have been already ignored in Stage~$1$. The remaining benchmarks, comprising semiorthogonal user group and random access, also report good satisfaction ratios, which demonstrate a gain of about $88\%$ and $85\%$, respectively. Overall, Algorithm~\ref{Algorithm2}, Relax, outperforms the baseline with with a slight increase of approximately $16.7\%$. It unveils that power allocation can improve all the benchmarks dramatically, so the gaps among them are reduced.  

Table~\ref{tab:table1} summarizes the benefits of power allocation compared with the fixed power level. All the considered benchmarks show up superior gains from utilizing the power allocation. The minimum improvement comes from the proposed method with the strict individual QoS requirements, which is $20\%$ for both the sum and per-user throughput. Random access has the most remarkable improvement on the sum throughput up to $47\%$ by carefully allocating the transmit power to each user instead of simply using a fixed power level. Meanwhile, semiorthogonal user group demonstrates an improvement of about $39\%$ and $44\%$ for the sum and per-user throughput, respectively. In addition, compared to the fixed power level, the served throughput improves from $31\%$ to $47\%$ by the optimized power under the given set of the parameter settings.

\vspace*{-0.25cm}
\section{Conclusion} \label{Sec:Conclusion}
Unlike previous works focused on user scheduling and power allocation for a particular time instance, this paper brought the same methodology but extended to multiple time slots. We proposed a user scheduling strategy for large-scale MB-HTS systems where many users simultaneously request to access the network. We formulated a total throughput optimization maximization problem in an observed window time subject to the individual QoS requirements. Due to the challenges in defeating the non-convexity, we proposed a heuristic algorithm to obtain a local solution with low computational complexity. The system can allocate radio resources to plenty of users. Numerical results manifested that all scheduled users have better QoSs than requested on average. Besides, the proposed algorithms offer higher sum throughput [Mbps] per time slot than the other benchmarks with about $20\%$.

Long-term satellite resource management is a promising research topic with many potential challenges. Each user can be scheduled multiple times in different time slots, resulting in instantaneous and aggregated throughput. Interesting future works should formulate and solve other optimization problems in satellite communications over an observed window time, such as beam hopping, gateway placement, and applying machine learning to reduce computational complexity and achieve better performance than traditional-based optimization approaches towards practical applications in real-time.
\vspace*{-0.25cm}
\appendix
\subsection{Useful Definitions} \label{Appendix:UsefulDef}
This appendix provides the useful definitions popularly utilized in handling geometric and sigomial programs. In particular,  Definition~\ref{Def1} perceives the perception of monomial, posynomial, and signomial functions, while Definition~\ref{Def2} provides the geometric program on the standard form.
\begin{definition} \label{Def1}
	Let us define $h(x_1, \ldots, x_{|\mathcal{K}(t)|})$ as a multivariate function of $x_1, \ldots, x_{|\mathcal{K}(t)|}$ as
	\begin{equation}
		h(x_1, \ldots, x_{|\mathcal{K}(t)|}) = c_m \prod\limits_{k \in \mathcal{K}(t)}   x_{k}^{b_{m,k}},
	\end{equation}
	then $h(x_1, \ldots, x_{|\mathcal{K}(t)|})$ is a monomial function if $c_m > 0$ and $b_{m,k}, \forall k,$ are real numbers.  As an extension, if the multivariate function $h(x_1, \ldots, x_{|\mathcal{K}(t)|})$ is defined by a summation of $\widetilde{M}$ terms as
	\begin{equation}
		h(x_1, \ldots, x_{|\mathcal{K}(t)|}) = \sum\limits_{m=1}^{\widetilde{M}}c_m \prod\limits_{k \in \mathcal{K}(t)}   x_{k}^{b_{m,k}}, 
	\end{equation}
	then $h(x_1, \ldots, x_{|\mathcal{K}(t)|})$ is a posynomial function if all $c_m, \forall m,$ are non-negative. Notice that $h(x_1, \ldots, x_{|\mathcal{K}(t)|})$ is a signomial function when at least one $c_m$ is negative.
\end{definition}
\begin{definition} \label{Def2}
	A geometric program is on standard form as
	\begin{equation} \label{Prob:geometricprog}
		\begin{aligned}
			& \underset{ \mathbf{x} \in \mathcal{X} }{\mathrm{maximize}}
			&&  f_0( \mathbf{x})  \\
			& \,\,\mathrm{subject \,to}
			&& f_k(\mathbf{x}) \leq 1, \forall k = 1, \ldots K,\\
			&& &   h_{k'} (\mathbf{x}) = 1, \forall k' = 1, \ldots, K',
		\end{aligned}
	\end{equation}
	where $\mathbf{x}$ is the optimization variable vector with the feasible domain $\mathcal{X}$. The objective function and the constraints fulfill the following conditions:
	\begin{itemize}
		\item The objective function $f_0 (\mathbf{x})$ is either a monomial or posynomial function. 
		\item The inequality constraints  $f_k(\mathbf{x}),\forall k,$ is either monomial or posynomial functions.
		\item The equality constraint  $h_{k'}(\mathbf{x})$ can be monomial functions.
	\end{itemize}
	Since a geometric program has a hidden convex structure, the globally optimal solution to problem~\eqref{Prob:geometricprog} can be obtained in polynomial time. Besides, \eqref{Prob:geometricprog} is a signomial program if at least one of those functions is  signomial, and therefore this problem is nonconvex.
\end{definition}
\subsection{Proof of Theorem~\ref{Theorem:SelectedUser}} \label{Appendix:SelectedUser}
We first prove that at the $t$-th outer iteration, the objective function of problem~\eqref{Problemv1} is non-decreasing along with inner iterations. Let us introduce a constant $\alpha_m^t$ with $ m \in \{|\mathcal{K}(t-1)|, \ldots, M\},$ as follows
\begin{equation}
	\alpha_m^t = \sum\limits_{k' \in \widetilde{\mathcal{K}}_m ^{\ast} (t) } R_{k'} (\widetilde{\mathcal{K}}_m ^{\ast} (t) ),
\end{equation}
then by exploiting \eqref{eq:Ser1}, the following series of inequality holds
\begin{equation}
	\alpha_M^t \geq \alpha_{M-1}^t \geq \ldots \geq \alpha_{|\mathcal{K}(t-1)|}^t,
\end{equation}
which demonstrates the non-decreasing property of the sum throughput in every time slot. Due to the non-negative property of the instantaneous channel capacity, we further obtain
\begin{equation}
	\sum\limits_{t'=1}^{t} \sum\limits_{k \in \mathcal{K}(t') } R_k(\mathcal{K}(t') ) \geq \sum\limits_{t'=1}^{t-1} \sum\limits_{k \in \mathcal{K}(t') } R_k(\mathcal{K}(t')),
\end{equation}
which manifests the fact that the objective function of problem~\eqref{Problemv1} is non-decreasing along with iterations. For a given set of transmit power coefficients, the instantaneous throughput of scheduled user is finite. Hence, the objective function of problem~\eqref{Problemv1} is upper bounded and Algorithm~\ref{Algorithm1} converges to a fixed point solution. Additionally, \eqref{eq:Ser2} ensures the individual QoS requirements and therefore we conclude the proof.
\vspace*{-0.25cm}
\subsection{Proof of Lemma~\ref{theorem:geometricprog}} \label{Appendix:geometricprog}
We first convert the SINR constraint of scheduled user~$k$ from a signomial to a posynomial. Mathematically, we reformulate the SINR constraint of this user to as 
\begin{equation} \label{eq:FirstSINR}
\begin{split}
	& \gamma_{k}(t)\sum\limits_{k' \in \mathcal{K}(t) \setminus \{ k\}} p_{k'}(t) \left| \mathbf{h}_{k}^H \mathbf{w}_{k'}(t) \right|^2 + \gamma_{k}(t) \sigma^2  \\
	&   \leq  \sum\limits_{k' \in \mathcal{K}(t) } p_{k'}(t) \left| \mathbf{h}_{k}^H \mathbf{w}_{k'}(t) \right|^2 + \sigma^2,
\end{split}
\end{equation}
then, in order to process further, we introduce a function $g_{k}(\{ p_{k'}(t) \} )$ to denote the right-hand side of \eqref{eq:FirstSINR} as
\begin{equation} \label{eq:gkt}
	g_{k}(\{ p_{k'}(t) \} ) = \sum\limits_{k' \in \mathcal{K}(t) } p_{k'}(t) \left| \mathbf{h}_{k}^H \mathbf{w}_{k'}(t) \right|^2 + \sigma^2.
\end{equation}
Utilizing the arithmetic mean-geometric mean inequality in Lemma~\ref{lemma:AMGM}, we lower bound $g_{k}(\{ p_{k'} (t) \} )$ by $\tilde{g}_{k}(\{ p_{k'}(t) \} )$ as
\begin{equation}
	g_{k}(\{ p_{k'}(t) \} ) \geq \tilde{g}_{k}(\{ p_{k'}(t) \} ),
\end{equation}
with the weights fulfilling the condition~\eqref{eq:WeightsSINRv1}. As a consequence, the constraint \eqref{eq:Ratev1} for scheduled user~$k$ with $k \in \mathcal{K}(t)$ is approximated to as shown in \eqref{eq:SINRConstraintv1}. It should be noticed that  \eqref{eq:SINRConstraintv1} is a posynomial constraint due to the left-hand side is a posynomial function. Problem~\eqref{Prob:SubProbv5} is hence a geometric program on standard form.
\subsection{Proof of Theorem~\ref{theorem:KKT}} \label{Appendix:KKT}
At the $t$-th time slot, we define the feasible domain of the signomial optimization problem~\eqref{Prob:SubProbv4} as
\begin{multline}
	\mathcal{P}(t) = \\	\Big\{ p_{k}(t), \forall k \in \mathcal{K}(t):  p_{k}(t)\in \mathbb{R}_{+}, \sum_{k \in \mathcal{K}(t)} p_{k}(t) \leq P_{\max} \Big\},
\end{multline}
which is compact set. The globally optimal solution to problem~\eqref{Prob:SubProbv5} obtained at the $i$-th iteration is included in the following set 
\begin{equation} \label{eq:OptPower}
\mathcal{I}^{(i)} (t) = \big\{ p_k^{\ast,(i)}(t), \forall k \in \mathcal{K}(t)\big\},
\end{equation}
then $\mathcal{I}^{(i)} (t) \subseteq \mathcal{P}(t)$. Let us denote $\tilde{g}_{k}^{(i)} \big(\{ p_{k'}^{(i)}(t) \} \big)$ the $\tilde{g}_{k} \big(\{ p_{k'}(t) \} \big)$ function at the $i$-th iteration. By utilizing \eqref{eq:gkm} and \eqref{eq:gkt}, we can construct the following properties
\begin{align}
g_{k}(\{ p_{k'}(t) \} ) &\geq \tilde{g}_{k} \big(\{ p_{k'}^{(i)}(t) \} \big), \label{eq:gki}\\
g_{k} \big(\{ p_{k'}^{\ast,(i)}(t) \} \big) &= \tilde{g}_{k}^{(i)} \big(\{ p_{k'}^{(i)}(t) \} \big), \label{eq:gkiv1} \\
\frac{\partial g_{k} \big(\{ p_{k'}^{\ast,(i)}(t) \} \big)}{\partial p_k^{\ast,(i)}(t)} &= \frac{\partial \tilde{g}_{k}^{(i+1)} \big(\{ p_{k'}^{(i)}(t) \} \big)}{\partial p_k^{\ast,(i)}(t)}. \label{eq:gkiv2}
\end{align}
The lower bound of $g_{k}(\{ p_{k'}(t) \} )$ in \eqref{eq:gki} applied to all the scheduled users in the $t$-th time slot indicates that the global optimum to the geometric program \eqref{Prob:SubProbv5} is a feasible point to the signomial program \eqref{Prob:SubProbv4} since
\begin{equation}
\frac{g_{k}(\{ p_{k'}(t) \} )}{D_k (\{ p_{k'}(t) \} )} \geq \frac{\tilde{g}_{k} \big(\{ p_{k'}^{(i)}(t) \} \big)}{D_k (\{ p_{k'}(t) \} )} \geq \gamma_k(t) -1, \forall k \in \mathcal{K}(t),
\end{equation}
where $D_k (\{ p_{k'}(t) \} ) \triangleq \sum\nolimits_{k' \in \mathcal{K}(t) \setminus \{ k\} } p_{k'} (t) \big|\mathbf{h}_{k}^H \mathbf{w}_{k'}(t)|^2 +\sigma^2$. Consequently, we can construct the following series of the inequalities
\begin{equation} \label{eq:Series}
\begin{split}
\ldots & \geq \widehat{\mathrm{SINR}}_k^{(i+1)} \big(\mathcal{I}^{(i)} (t)  \big) \stackrel{(a)}{=} \mathrm{SINR}_k \big(\mathcal{I}^{(i)} (t) \big) \\
&\stackrel{(b)}{\geq} \widehat{\mathrm{SINR}}_k^{(i)} \big(\mathcal{I}^{(i)} (t) \big) \stackrel{(c)}{\geq} \widehat{\mathrm{SINR}}_k^{(i)}\big(\mathcal{I}^{(i-1)} (t) \big) = \ldots,
\end{split}
\end{equation}
where $\mathrm{SINR}_k \big(\mathcal{I}^{(i)} (t) \big)$ is the instantaneous SINR value of user~$k$ defined in \eqref{eq:SINRk} with the optimal power coefficients in \eqref{eq:OptPower}, while
\begin{equation}
\widehat{\mathrm{SINR}}_k (t) \triangleq \frac{\tilde{g}_{k} \big(\{ p_{k'}(t) \} \big)}{D_k (\{ p_{k'}(t) \} )},
\end{equation}
is the approximate SINR expression by using the arithmetic mean-geometric mean inequality in Lemma~\ref{lemma:AMGM}. In \eqref{eq:Series}, $(a)$ is obtained by using \eqref{eq:gkiv1}, which indicates that the approximate and original SINR values are equal to each other at the global optimum to problem~\eqref{Prob:SubProbv5}; $(b)$ is because of \eqref{eq:gki} meaning that the optimal solution to \eqref{Prob:SubProbv5} is always feasible to the original problem~\eqref{Prob:SubProbv4}; and $(c)$ is obtained by the fact that we enable to obtain the globally optimal solution to problem  \eqref{Prob:SubProbv4} thanks to the hidden convex structure. The series of inequalities in \eqref{eq:Series} indicates that 
\begin{equation} \label{eq:gammakit}
 \prod\limits_{k \in \mathcal{K}(t) }   \gamma_{k}^{(i+1), \ast}(t)  \geq \prod\limits_{k \in \mathcal{K}(t) }   \gamma_{k}^{(i), \ast}(t), 
\end{equation}
where $\gamma_{k}(t)^{(i),\ast}$ is the globally optimal solution to problem~\eqref{Prob:SubProbv5} at the $i$-th iteration. Consequently, the objective function is non-decreasing along with iterations. By virtue of the limited power budget constraint \eqref{eq:PowerBudget1},  the SINR values are bounded from above, i.e., 
\begin{equation}
\widehat{\mathrm{SINR}}_k(\{ p_{k'}(t)\}), \mathrm{SINR}_k(\{ p_{k'}(t)\}) < \infty, \forall k \in \mathcal{K}(t), 
\end{equation}
which ensure that solving \eqref{Prob:SubProbv5} and updating the weight values as in \eqref{eq:gkm} converge after a finite number of iterations. If the convergence holds at the $i$-th iteration, i.e.,
\begin{equation} \label{eq:Eqv}
\prod\limits_{k \in \mathcal{K}(t) }   \gamma_{k}^{(i+1), \ast}(t)  = \prod\limits_{k \in \mathcal{K}(t) }   \gamma_{k}^{(i), \ast}(t), 
\end{equation}
then the global optimal solution set $ \mathcal{I}^{(i)} (t)$ should be a solution at the $(i+1)$-th iteration. Otherwise, \eqref{eq:Eqv} does not hold. Since the feasible domain of problem~\eqref{Prob:SubProbv4} is compact, the Slater's condition is satisfied and \eqref{eq:gki}--\eqref{eq:gkiv2} ensure that the KKT conditions coincide for both problems \eqref{Prob:SubProbv4} and \eqref{Prob:SubProbv5}. The solution obtained by Algorithm~\ref{Algorithm2} should be a KKT point to problem~\eqref{ProblemvData}, as stated in the theorem.
\vspace*{-0.25cm}
\bibliographystyle{IEEEtran} 
\bibliography{IEEEabrv,refs}
\end{document}